\documentclass{svmult}

\usepackage{makeidx}         
\usepackage{graphicx}        
\usepackage{amssymb,amstext,amsmath,amsfonts}
\usepackage{psfrag,epsfig,multimedia}
\usepackage{graphicx}
\usepackage{turnstile}

\usepackage{tikz}
\tikzstyle{matched}=[circle, draw, fill=black!30,inner sep=0pt, minimum width=4pt,fill=black!100]
\usepackage{pgf,tikz}
\usetikzlibrary{arrows}

\usepackage{graphicx}
\usepackage{color}
\usepackage{float}
\usepackage{tabularx}
\usepackage{hero}

\def\mU{{\mathbb U}}

\def\Beta{B}

\newtheorem{thm}{Theorem}
\newtheorem{lem}{Lemma}  
\newtheorem{coro}{Corollary}
\usepackage{multicol}        
\usepackage[bottom]{footmisc}

\definecolor{xdxdff}{rgb}{0.49,0.49,1}
\definecolor{ffffff}{rgb}{1,1,1}
\definecolor{qqqqff}{rgb}{0,0,1}
\definecolor{ffqqqq}{rgb}{1,0,0}

\makeindex             

\begin{document}

\title*{Spectral Correlation Hub Screening of Multivariate Time Series}
\author{Hamed Firouzi, Dennis Wei \and Alfred O. Hero III}
\institute{Electrical Engineering and Computer Science Department, University of Michigan, USA
\texttt{firouzi@umich.edu, dlwei@eecs.umich.edu, hero@eecs.umich.edu}}
%
%
\maketitle


\begin{abstract}
This chapter discusses correlation analysis of stationary multivariate Gaussian 
time series in the spectral or Fourier domain. The goal is to identify the hub time series, i.e., those that are highly correlated with a specified number of other time series.  We show that Fourier components of the time series at different frequencies are asymptotically statistically independent.  This property permits independent correlation analysis at each frequency, alleviating the computational and statistical challenges of high-dimensional time series. 
To detect correlation hubs at each frequency, 
an existing correlation screening method is extended to the complex numbers to accommodate complex-valued Fourier components.  We characterize the number of hub discoveries at specified correlation and degree thresholds in the regime of increasing dimension and fixed sample size.  The theory specifies appropriate thresholds 
to apply to sample correlation matrices 
to detect hubs and also allows statistical significance to be attributed to hub discoveries. 
Numerical results illustrate the accuracy of the theory 
and the usefulness of the proposed spectral framework.
\end{abstract}

\section{Introduction}
Correlation analysis of multivariate time series is important in many applications
such as 
wireless sensor networks, computer networks, neuroimaging, and finance \cite{vuran2004spatio,paffenroth2013space,friston2011statistical,zhang2003correlation,tsay2005analysis}.
%
This chapter focuses on the problem of detecting 
\emph{hub} time series, ones that have a high degree of interaction with other time series as measured by correlation or partial correlation. 
Detection of hubs can lead to reduced 
computational and/or sampling costs. 
For example in wireless sensor networks, the identification 
of hub nodes can 
be useful for reducing power usage and adding or removing sensors from the network \cite{stanley2012intelligent,li2008wireless}.
Hub detection can also give 
new insights about underlying structure in the dataset. 
In neuroimaging for instance, 
studies have consistently shown the existence of highly connected hubs in brain graphs (connectomes) \cite{bullmore2009complex}.
In finance, a hub might indicate a vulnerable financial instrument or a sector whose collapse could have a major effect on the market \cite{hero2012hub}. 

Correlation analysis becomes challenging for multivariate time series 
when the dimension $p$ of the time series, i.e.\ the number of scalar time series, 
and the 
number of time samples $N$ 
are large \cite{zhang2003correlation}. 
A naive approach is to treat the time series as a set of independent samples of a $p$-dimensional random vector 
and estimate the associated covariance or correlation matrix, 
but this approach completely ignores temporal correlations 
as it only considers dependences 
at the same time instant and not between different time instants. 
The work in \cite{chen2013covariance} accounts for temporal correlations by quantifying their effect on convergence rates in covariance and precision matrix estimation; however, only correlations at the same time instant are estimated.  
A more general approach is to consider all correlations between any two time instants 
of any two series within a window of $n \leq N$ consecutive samples, where the previous case corresponds to $n = 1$. 
However, in general this would entail the estimation of an $np \times np$ correlation matrix from a reduced sample of size $m = N/n$, which can be computationally costly as well as statistically problematic. 

In this chapter, we propose \emph{spectral} correlation analysis as a method of overcoming the issues discussed above.  As before, the time series are divided into $m$ temporal segments of $n$ consecutive samples, but instead of estimating temporal correlations directly, the method performs 
analysis on the Discrete Fourier Transforms (DFT) of the time series. 
We prove in Theorem \ref{thm:independence} that for stationary, jointly Gaussian 
time series under the mild condition of absolute summability of the auto- and cross-correlation functions, different Fourier components (frequencies) become asymptotically independent of each other as the DFT length $n$ 
increases. This 
property of stationary 
Gaussian processes allows us to focus on the $p\times p$ correlations 
at each frequency separately without having to consider 
correlations between different frequencies. 
Moreover, spectral analysis 
isolates correlations at 
specific frequencies or timescales, potentially leading to greater insight. 
To make aggregate inferences based on all frequencies, straightforward procedures for multiple inference can be used as described in Section \ref{sec:MultipleInf}.

The spectral approach reduces the detection of hub time series to the independent detection of hubs at each frequency.  However, in exchange for achieving spectral resolution, the sample size is reduced by the factor $n$, from $N$ to $m = N/n$. 
To confidently detect hubs in this high-dimensional, low-sample regime (large $p$, small $m$), as well as to accommodate complex-valued DFTs, we develop 
a method that 
we call \emph{complex-valued (partial) correlation screening}. 
This 
is a generalization of the correlation and partial correlation screening method of \cite{hero2011large,hero2012hub,firouzi2013predictive} to 
complex-valued random variables.  For each frequency, the method computes the sample (partial) correlation matrix of the DFT components of the $p$ time series. 
Highly correlated variables (hubs) are then identified by thresholding the sample correlation matrix at a level $\rho$ and screening for rows (or columns) with a specified number $\delta$ of 
non-zero entries.  

We characterize the behavior of complex-valued correlation screening in the high-dimensional regime of large $p$ and fixed sample size $m$.  Specifically, Theorem \ref{prop:parcor} and Corollary \ref{prop:parcor1} give asymptotic expressions in the limit $p \to \infty$ 
for the mean number of hubs detected at thresholds $\rho, \delta$ and the probability of discovering at least one such hub.  Bounds on the rates of convergence are also provided.  These results show that the number of hub discoveries undergoes a phase transition as $\rho$ decreases from $1$, from almost no discoveries to the maximum number, $p$.  An expression \eqref{CriticalThreshold} for the critical threshold $\rho_{c,\delta}$ is derived to guide the selection of $\rho$ under different settings of $p$, $m$, and $\delta$.  Furthermore, given a null hypothesis that the population correlation matrix is sufficiently sparse, 
the expressions in Corollary \ref{prop:parcor1} become independent of the underlying probability distribution and can thus be easily evaluated.  This allows the statistical significance of a hub discovery to be quantified, specifically in the form of a $p$-value under the null hypothesis.  We note that our results on complex-valued correlation screening apply more generally than to spectral correlation analysis and thus may be of independent interest. 

The remainder of the chapter is organized as follows. 
Section \ref{GaussianWSS} presents notation and definitions for multivariate time series and establishes the asymptotic independence of spectral components.  Section \ref{sec:corrScreen} describes complex-valued correlation screening and characterizes its properties in terms of numbers of hub discoveries and phase transitions. 
Section \ref{sec:MultipleInf} discusses the application of complex-valued correlation screening to the spectra of multivariate time series. 
Finally, Sec. \ref{sec:Sims} illustrates the applicability of the proposed framework through simulation analysis. 


\subsection{Notation}
A triplet $(\Omega,\mathcal{F},\mathbb{P})$ represents a probability space with sample space $\Omega$, $\sigma$-algebra of events $\mathcal{F}$, and probability measure $\mathbb{P}$. For an event $A \in \mathcal{F}$, $\mathbb{P}(A)$ represents the probability of $A$. Scalar random variables and their realizations are denoted with upper case and lower case letters, respectively. Random vectors and their realizations are denoted with bold upper case and bold lower case letters. The expectation operator is denoted as $\mathbb E$. 
For a random variable $X$, the cumulative probability distribution (cdf) of $X$ is defined as $F_{X}(x) = \mathbb{P}(X \leq x)$. For an absolutely continuous cdf $F_{X}(.)$ the probability density function (pdf) is defined as $f_{X}(x) = dF_{X}(x)/dx$. The cdf and pdf 
are defined similarly for random vectors. Moreover, we follow the definitions in \cite{durrett2010probability} for conditional probabilities, conditional expectations and conditional densities.


For a complex number $z = a+b\sqrt{-1} \in \mathbb{C}$, $\Re(z)=a$ and $\Im(z)=b$ represent the real and imaginary parts of $z$, respectively. A complex-valued random variable is composed of 
two real-valued random variables as its real and imaginary parts.  A complex-valued Gaussian variable 
has real and imaginary parts that are Gaussian. A complex-valued (Gaussian) random vector is a vector whose entries are complex-valued (Gaussian) random variables. 
The covariance of a $p$-dimensional complex-valued random vector $\bY$ and a $q$-dimensional complex-valued random vector $\bZ$ is a $p \times q$ matrix defined as
\[
\cov(\bY, \bZ) = \mathbb E\left[ (\bY - \mathbb E[\bY]) (\bZ - \mathbb E[\bZ])^{H} \right],
\]
where $^{H}$ denotes the Hermitian transpose.  We write $\cov(\bY)$ for $\cov(\bY,\bY)$ and $\var(Y) = \cov(Y,Y)$ for the variance of a scalar random variable $Y$.  The correlation coefficient between random variables $Y$ and $Z$ is defined as 
\[
\mathrm{cor}(Y,Z) = \frac{\cov(Y,Z)}{\sqrt{\var(Y) \var(Z)}}.
\]

Matrices are also denoted by bold upper case letters. In most cases the distinction between matrices and random vectors will be clear from the context. For a matrix $\bA$ we represent the $(i,j)$th entry of $\bA$ by $a_{ij}$. 
Also $\bD_{\bA}$ represents the diagonal matrix that is obtained by zeroing out all but the diagonal entries of $\bA$.

\section{Spectral Representation of Multivariate Time Series}
\label{GaussianWSS}

\subsection{Definitions}

Let $\bX(k) = [X^{(1)}(k), X^{(2)}(k), \cdots X^{(p)}(k)]$, $k \in \mathbb{Z}$, be a multivariate time series with time index $k$. 
We assume that the time series $X^{(1)}, X^{(2)}, \cdots X^{(p)}$ are second-order stationary 
random processes, i.e.:
\be
\mathbb E[X^{(i)}(k)] = \mathbb E[X^{(i)}(k+\Delta)]
\label{eq:transinv1}
\ee
and
\be
{\bf \cov} [X^{(i)}(k),X^{(j)}(l)] = {\bf \cov} [X^{(i)}(k+\Delta),X^{(j)}(l+\Delta)]
\label{eq:transinv2}
\ee
for any integer time shift $\Delta$.

For $1 \leq i \leq p$, let $\bX^{(i)} = [X^{(i)}(k),\cdots,X^{(i)}(k+n-1)]$ denote any vector of $n$ consecutive samples of time series $X^{(i)}$. The $n$-point Discrete Fourier Transform (DFT) of $\bX^{(i)}$ is denoted by $\bY^{(i)} = [Y^{(i)}(0),\cdots,Y^{(i)}(n-1)]$ and defined by 
\ben
\bY^{(i)}=\bW \bX^{(i)},~~~ 1 \leq i \leq p
\een
in which $\bW$ is the DFT matrix:
\ben
\bW = \frac{1}{\sqrt{n}}
 \begin{bmatrix}
  1 & 1 & \cdots & 1 \\
  1 & \omega & \cdots & \omega^{n-1} \\
  \vdots  & \vdots  & \ddots & \vdots  \\
  1 & \omega^{n-1} & \cdots & \omega^{(n-1)^2}
 \end{bmatrix},
\een
where $\omega = e^{-2\pi\sqrt{-1}/n}$.

We denote the $n \times n$ population covariance matrix of $\bX^{(i)}$ as $\bC^{(i,i)}=[c^{(i,i)}_{kl}]_{1 \leq k,l \leq n}$ and the $n \times n$ population cross covariance matrix between $\bX^{(i)}$ and $\bX^{(j)}$ as $\bC ^{(i,j)}=[c^{(i,j)}_{kl}]_{1\leq k,l \leq n}$ for $i \neq j$. The translation invariance properties \eqref{eq:transinv1} and \eqref{eq:transinv2} imply that $\bC^{(i,i)}$ and $\bC^{(i,j)}$ are Toeplitz matrices. Therefore  $c^{(i,i)}_{kl}$ and $c^{(i,j)}_{kl}$ depend on $k$ and $l$ only through the quantity $k-l$.  
Representing the $(k,l)$th entry of a Toeplitz matrix $\bT$ by $t(k-l)$, we write 
\ben
c^{(i,i)}_{kl} = c^{(i,i)}(k-l)
~~~\text{and}~~~
c^{(i,j)}_{kl} = c^{(i,j)}(k-l),
\een
where $k-l$ takes values from $1-n$ to $n-1$. In addition, $\bC^{(i,i)}$ is symmetric. 

\subsection{Asymptotic Independence of Spectral Components}

The following theorem states that for stationary 
time series, 
DFT components at different spectral indices (i.e.\ frequencies) 
are asymptotically uncorrelated under the condition that the auto-covariance 
and cross-covariance 
functions are absolutely summable. This theorem follows directly from the spectral theory of large Toeplitz matrices, see, for example, \cite{grenander1958toeplitz} and \cite{Gray06}. However, for the benefit of the reader we give a self contained proof of the theorem.

\begin{thm}
Assume $\lim_{n \rightarrow \infty} \sum_{t=0}^{n-1} |c^{(i,j)}(t)| = M^{(i,j)}< \infty$ for all $1 \leq i,j \leq p$. Define $\mathrm{err}^{(i,j)}(n) = M^{(i,j)} - \sum_{m'=0}^{n-1} |c^{(i,j)}(m')|$ and $\mathrm{avg}^{(i,j)}(n) = \frac{1}{n} \sum_{m'=0}^{n-1} \mathrm{err}^{(i,j)}(m')$. Then for $k \neq l$, we have:
\ben
\mathrm{cor}\left(Y^{(i)}(k),Y^{(j)}(l)\right) = O(\max \{ 1/n, \mathrm{avg}^{(i,j)}(n)  \}).
\een
In other words $Y^{(i)}(k)$ and $Y^{(j)}(l)$ are asymptotically uncorrelated 
as $n \to \infty$.
\label{thm:independence}
\end{thm}
\begin{proof}
Without loss of generality we assume that the time series have zero mean (i.e. $\mathbb E[X^{(i)}(k)] = 0, 1 \leq i \leq p, 0 \leq k \leq n-1$). We first establish a representation of $\mathbb E[Z^{(i)}(k)Z^{(j)}(l)^*]$ for general linear functionals:
\ben
Z^{(i)}(k) = \sum_{m'=0}^{n-1} g_k(m')X^{(i)}(m'),
\een
in which $g_k(.)$ is an arbitrary complex sequence for $0 \leq k\leq n-1$. We have:
\be
&& \mathbb E[Z^{(i)}(k)Z^{(j)}(l)^*] \nonumber \\ &=& \mathbb E\left[\left( \sum_{m'=0}^{n-1} g_k(m')X^{(i)}(m') \right)\left(\sum_{n'=0}^{n-1} g_l(n')X^{(j)}(n')\right)^*\right] \nonumber \\
&=& \sum_{m'=0}^{n-1} g_k(m') \sum_{n'=0}^{n-1} g_l(n')^* \mathbb E[X^{(i)}(m') X^{(j)}(n')^*] \nonumber \\
&=&\sum_{m'=0}^{n-1} g_k(m') \sum_{n'=0}^{n-1} g_l(n')^* c^{(i,j)}_{m'n'} \label{CovExpr}
\ee
Now for a Toeplitz matrix $\bT$, define the circulant matrix $\bD_{\bT}$ as:
\ben
\bD_{\bT} = 
 \begin{bmatrix}
  t (0)& t(-1)+t(n-1) & \cdots & t(1-n)+t(1) \\
  t(1)+t(1-n) & t(0) & \cdots & t(2-n) +t(2) \\
  \vdots  & \vdots  & \ddots & \vdots  \\
  t(n-2) + t(-2) ~& ~~t(n-3) + t(-3) ~&~ \cdots & ~t(-1) + t(n-1) \\
  t(n-1) + t(-1) ~& ~~t(n-2) + t(-2) ~&~ \cdots & t(0)
 \end{bmatrix}
\een
We can write:
\ben
\bC^{(i,j)} = \bD_{\bC^{(i,j)}} + \bE^{(i,j)}
\label{eq:CircApprox}
\een
for some Toeplitz matrix $\bE^{(i,j)}$. Thus $c^{(i,j)}(m'-n') = d^{(i,j)}(m'-n') + e^{(i,j)}(m'-n')$ where $d^{(i,j)}(m'-n')$ and $e^{(i,j)}(m'-n')$ are the $(m',n')$ entries of $\bD_{\bC^{(i,j)}}$ and $\bE^{(i,j)}$, respectively. Therefore, \eqref{CovExpr} can be written as:
\ben
\sum_{m'=0}^{n-1} g_k(m') \sum_{n'=0}^{n-1} g_l(n')^*  d^{(i,j)}(m'-n') + \sum_{m'=0}^{n-1}\sum_{n'=0}^{n-1} g_k(m')  g_l(n')^* e^{(i,j)}(m'-n')
\een
The first term can be written as:
\ben
\sum_{m'=0}^{n-1} g_k(m') \left(g_l^*  \circledast d^{(i,j)}\right) (m') = \sum_{m'=0}^{n-1} g_k(m') v^{(i,j)}_l(m')
\een
where we have recognized $v^{(i,j)}_l(m')= g_l^*  \circledast d^{(i,j)}$ as the circular convolution of $g^*_l(.)$ and $d^{(i,j)}(.)$ \cite{oppenheim1989discrete}. Let $G_k(.)$ and $D^{(i,j)}(.)$ be the the DFT of $g_k(.)$ and $d^{(i,j)}(.)$, respectively. By Plancherel's theorem \cite{conway1990course} we have:
\be
\sum_{m'=0}^{n-1} g_k(m') v^{(i,j)}_l(m')
 &=&
\sum_{m'=0}^{n-1} g_k(m') \left(v^{(i,j)}_l(m')^* \right)^* \nonumber \\ &=& \sum_{m'=0}^{n-1} G_k(m') \left( G_l(m')D^{(i,j)}(-m')^* \right)^* \nonumber \\ &=& \sum_{m'=0}^{n-1} G_k(m') G_l(m')^* D^{(i,j)}(-m'). \label{firsttrem}
\ee
Now let $g_k(m')=\omega^{km'}/\sqrt{n}$ for $0 \leq k,m' \leq n-1$. For this choice of $g_k(.)$ we have $G_k(m') = 0$ for all $m' \neq n-k$ and $G_k(n-k)=1$. Hence for $k \neq l$ the quantity \eqref{firsttrem} becomes $0$. Therefore using the representation $\bE^{(i,j)} = \bC^{(i,j)} - \bD_{\bC^{(i,j)}}$ we have:
\be
\label{eq:Ycov_bound}
|\cov\left(Y^{(i)}(k),Y^{(j)}(l)\right)|  &=& |\mathbb E[Y^{(i)}(k)Y^{(j)}(l)^*]|  \nonumber \\ &=& |\sum_{m'=0}^{n-1}\sum_{n'=0}^{n-1} g_k(m')  g_l(n')^* e^{(i,j)}(m'-n')|  \nonumber \\ &\leq&
\frac{1}{n} \sum_{m'=0}^{n-1}\sum_{n'=0}^{n-1} |e^{(i,j)}(m'-n')| \nonumber \\ &=& \frac{2}{n} \sum_{m'=0}^{n-1} m'|c^{(i,j)}(m')|,
\ee
in which the last equation is due to the fact that $|c^{(i,j)}(-m')| = |c^{(i,j)}(m')|$. 

Now using \eqref{firsttrem} and \eqref{eq:Ycov_bound} we obtain expressions for $\mathrm{var}\left(Y^{(i)}(k)\right)$ and $\mathrm{var}\left(Y^{(j)}(l)\right)$. Letting $j=i$ and $l=k$ in \eqref{firsttrem} and \eqref{eq:Ycov_bound} gives:
\be
&& \mathrm{var}\left(Y^{(i)}(k)\right) = \cov\left(Y^{(i)}(k),Y^{(i)}(k)\right) \nonumber  \\ \nonumber &=& \sum_{m'=0}^{n-1} G_k(m') G_k(m')^* D^{(i,i)}(-m') + \sum_{m'=0}^{n-1}\sum_{n'=0}^{n-1} g_k(m')  g_k(n')^* e^{(i,i)}(m'-n') \\ &=& n . \frac{1}{\sqrt{n}} . \frac{1}{\sqrt{n}}  D^{(i,i)}(k)  + \sum_{m'=0}^{n-1}\sum_{n'=0}^{n-1} g_k(m')  g_k(n')^* e^{(i,i)}(m'-n') \nonumber  \\ &=&  D^{(i,i)}(k)  + \sum_{m'=0}^{n-1}\sum_{n'=0}^{n-1} g_k(m')  g_k(n')^* e^{(i,i)}(m'-n'),
\label{eq:Var_i}
\ee
in which the magnitude of the summation term is bounded as:
\be
&&|\sum_{m'=0}^{n-1}\sum_{n'=0}^{n-1} g_k(m')  g_k(n')^* e^{(i,i)}(m'-n')|  \nonumber \\ &\leq&
\frac{1}{n} \sum_{m'=0}^{n-1}\sum_{n'=0}^{n-1} |e^{(i,i)}(m'-n')| \nonumber \\ &=& \frac{2}{n} \sum_{m'=0}^{n-1} m'|c^{(i,i)}(m')|.
\label{eq:Var_i_err_bound}
\ee
Similarly:
\be
\mathrm{var}\left(Y^{(j)}(l)\right) =  D^{(j,j)}(l)  + \sum_{m'=0}^{n-1}\sum_{n'=0}^{n-1} g_l(m')  g_l(n')^* e^{(j,j)}(m'-n'),
\label{eq:Var_j}
\ee
in which
\be
&&|\sum_{m'=0}^{n-1}\sum_{n'=0}^{n-1} g_l(m')  g_l(n')^* e^{(j,j)}(m'-n')|  \nonumber \\ &\leq&
\frac{2}{n} \sum_{m'=0}^{n-1} m'|c^{(j,j)}(m')|.
\label{eq:Var_j_err_bound}
\ee

To complete the proof the following lemma is needed.
\begin{lem}
\label{lem:convergence}
If $\{a_{m'}\}_{m'=0}^{\infty}$ is a sequence of non-negative numbers such that $\sum_{m'=0}^{\infty} a_{m'} = M < \infty$. Define $\mathrm{err}(n) = M - \sum_{m'=0}^{n-1} a_{m'}$ and $\mathrm{avg}(n) = \frac{1}{n} \sum_{m'=0}^{n-1} \mathrm{err}(m')$. Then $ |\frac{1}{n} \sum_{m'=0}^{n-1} m'a_{m'}|  \leq  M/n + \mathrm{err}(n) + \mathrm{avg}(n)$.
\end{lem}
\begin{proof}
Let $S_0=0$ and for $n \geq 1$ define $S_n = \sum_{m'=0}^{n-1} a_{m'}$. We have:
\ben
\sum_{m'=0}^{n-1} ma_{m'} = (n-1)S_{n} - (S_0 + S_1 + \ldots + S_{n-1}).
\een
Therefore:
\ben
\frac{1}{n}\sum_{m'=0}^{n-1} m'a_{m'} = \frac{n-1}{n} S_{n-1} - \frac{1}{n} \sum_{m'=0}^{n-1} S_{m'}.
\label{eq:avg_mul_sum}
\een
Since $ M - M/n - \mathrm{err}(n) \leq \frac{n-1}{n}S_{n-1} \leq M$ and $M - \mathrm{avg}(n) \leq \frac{1}{n} \sum_{m'=0}^{n-1} S_{m'} \leq M$, using the triangle inequality the result follows.
\qed
\end{proof}
Now let $a_{m'} = |c^{(i,j)}(m')|$. By assumption $\lim_{n \rightarrow \infty} \sum_{m'=0}^{n-1} a_{m'} = M^{(i,j)} < \infty$. Therefore, Lemma \ref{lem:convergence} along with \eqref{eq:Ycov_bound} concludes:
\be
\cov\left(Y^{(i)}(k),Y^{(j)}(l)\right) = O(\max \{ 1/n, \mathrm{err^{(i,j)}}(n), \mathrm{avg^{(i,j)}}(n)  \}).
\label{eq:cov_bound}
\ee
$\mathrm{err}^{(i,j)}(n)$ is a decreasing decreasing function of $n$. Therefore $\mathrm{avg}^{(i,j)}(n) \geq \mathrm{err}^{(i,j)} (n)$, for $n \geq 1$. Hence:
\ben
\cov\left(Y^{(i)}(k),Y^{(j)}(l)\right) = O(\max \{ 1/n, \mathrm{avg}^{(i,j)}(n)  \}).
\een
Similarly using Lemma \ref{lem:convergence} along with \eqref{eq:Var_i}, \eqref{eq:Var_i_err_bound}, \eqref{eq:Var_j} and \eqref{eq:Var_j_err_bound} we obtain:
\be
|\mathrm{var}\left(Y^{(i)}(k)\right) - D^{(i,i)}(k)| =  O(\max \{ 1/n, \mathrm{avg}^{(i,i)}(n)  \}),
\label{eq:Var_i_conv}
\ee
and
\be
|\mathrm{var}\left(Y^{(j)}(l)\right) - D^{(j,j)}(l)| =  O(\max \{ 1/n, \mathrm{avg}^{(j,j)}(n)  \}).
\label{eq:Var_j_conv}
\ee
Using the definition 
\ben
\mathrm{cor}\left(Y^{(i)}(k),Y^{(j)}(l)\right) = \frac{\cov\left(Y^{(i)}(k),Y^{(j)}(l)\right)} {\sqrt{\mathrm{var}\left(Y^{(i)}(k)\right)} \sqrt{ \mathrm{var}\left(Y^{(j)}(l)\right)}},
\een
and the fact that as $n \rightarrow \infty$, $D^{(i,i)}(k)$ and $D^{(j,j)}(l)$ converge to constants $\bC^{(i,i)}(k)$ and $\bC^{(j,j)}(l)$, respectively, equations \eqref{eq:cov_bound}, \eqref{eq:Var_i_conv} and \eqref{eq:Var_j_conv} conclude:
\ben
\mathrm{cor}\left(Y^{(i)}(k),Y^{(j)}(l)\right) = O(\max \{ 1/n, \mathrm{avg}^{(i,j)}(n)  \}).
\een
\flushright\qed
\end{proof}

As an example we apply Theorem \ref{thm:independence} 
to a scalar auto-regressive (AR) process $X(k)$ specified by 
\ben
X(k) = \sum_{l=1}^{L} \varphi_l X(k-l) + \varepsilon(k),
\een
in which $\varphi_{l}$ are real-valued coefficients and $\varepsilon(.)$ is a stationary 
process with no temporal correlation. The auto-covariance function of an AR process can be written as \cite{hamilton1994time}:
\ben
c(t) = \sum_{l=1}^L \alpha_l r_l^{|t|},
\een
in which $r_1,\ldots,r_l$ are the roots of the polynomial $\beta(x) = x^L - \sum_{l=1}^L \varphi_l x^{L-l}$. It is known that for a stationary AR process, $|r_l| < 1$ for all $1 \leq l \leq L$ \cite{hamilton1994time}. Therefore, using the definition of $\mathrm{err(.)}$ we have:
\ben 
\mathrm{err}(n) &=& \sum_{t=n}^{\infty}|c(t)|  = \sum_{t=n}^{\infty} |\sum_{l=1}^L \alpha_l r_l^{t}|  \leq  \sum_{l=1}^L |\alpha_l| \sum_{t=n}^\infty |r_l|^t \\ &=&  \sum_{l=1}^L |\alpha_l|\frac{|r_l|^n}{1-|r_l|} \leq C \zeta^n,
\een
in which $C = \sum_{l=1}^L |\alpha_l|/(1-|r_l|)$ and $\zeta = \max_{1 \leq l \leq L} |r_l| < 1$. Hence:
\ben
\mathrm{avg}(n) = \frac{1}{n} \sum_{m'=0}^{n-1} \mathrm{err}(m') \leq \frac{1}{n} \sum_{m'=0}^{n-1} C \zeta^{m'} \leq \frac{C}{n(1-\zeta)}.
\een
Therefore, Theorem \ref{thm:independence} concludes:
\ben
\mathrm{cor}\left( Y(k),Y(l) \right) = O(1/n), ~~~~~ k \neq l,
\een
where $Y(.)$ represents the $n$-point DFT of the AR process $X(.)$.

In the sequel, we assume that the time series $\bX$ is multivariate Gaussian, i.e., $X^{(1)},\ldots,X^{(p)}$ are jointly Gaussian processes.  It follows that the DFT components $Y^{(i)}(k)$ are jointly (complex) Gaussian as linear functionals of $\bX$.  Theorem \ref{thm:independence} then immediately implies asymptotic independence of DFT components through a well-known property of jointly Gaussian random variables. 

\begin{coro}\label{cor:independence}
Assume that the time series $\bX$ is multivariate Gaussian. Under the absolute summability conditions in Theorem \ref{thm:independence}, the DFT components $Y^{(i)}(k)$ and $Y^{(j)}(l)$ are asymptotically independent for $k \neq l$ and $n \to \infty$. 
\end{coro}

Corollary \ref{cor:independence} implies that for large $n$, correlation analysis of the time series $\bX$ can be done independently on each frequency in the spectral domain.  This reduces the problem of screening for hub time series to screening for hub variables among the $p$ DFT components at a given frequency.  A procedure for the latter problem and a corresponding theory are described next.


\section{Complex-Valued Correlation Hub Screening}
\label{sec:corrScreen}

This section discusses complex-valued correlation hub screening, a generalization of real-valued correlation screening in \cite{hero2011large,hero2012hub}, for identifying highly correlated components of a complex-valued random vector from its sample values.  The method is applied to multivariate time series in Section \ref{sec:MultipleInf} to discover correlation hubs among the spectral components at each frequency.  Sections \ref{subsec:corrScreenModel} and \ref{subsec:corrScreenProc} describe the underlying statistical model and the screening procedure.  Sections \ref{subsec:corrScreenUScore} and \ref{Props} provide background on the U-score representation of correlation matrices and associated definitions and properties.  Section \ref{sec:Theoretical} contains the main theoretical result characterizing the number of hub discoveries in the high-dimensional regime, while Section \ref{subsec:corrScreenPhase} elaborates on the phenomenon of phase transitions in the number of discoveries.

\subsection{Statistical Model}
\label{subsec:corrScreenModel}

We use the generic notation 
$\bZ=[Z_1,Z_2,\cdots,Z_p]^T$ in this section to refer to a complex-valued random vector.  The mean of $\bZ$ is denoted as $\boldsymbol \mu$ and its $p \times p$ non-singular covariance matrix is denoted as $\mathbf \Sigma$. 
We assume
that the vector $\bZ$ follows a complex elliptically contoured distribution with 
pdf $f_{\bZ}(\bz)=
g\left((\bz-\boldsymbol \mu)^H {\mathbf \Sigma}^{-1} (\bz-\boldsymbol
\mu)\right)$, in which $g: \mathbb{R}^{\geq 0} \rightarrow \mathbb{R}^{>0}$ is an integrable and strictly decreasing function \cite{micheas2006complex}.  This assumption generalizes the Gaussian assumption made in Section \ref{GaussianWSS} as the Gaussian distribution is one example of an elliptically contoured distribution. 

In correlation hub screening, the quantities of interest are the correlation matrix and partial correlation matrix associated with $\bZ$. 
These are defined as $\mathbf \Gamma = \bD_{\mathbf{\Sigma}}^{-\frac{1}{2}} \mathbf{\Sigma} \bD_{\mathbf{\Sigma}}^{-\frac{1}{2}}$
and $\mathbf \Omega = \bD_{\mathbf{\Sigma^{-1}}}^{-\frac{1}{2}} \mathbf{\Sigma^{-1}} \bD_{\mathbf{\Sigma^{-1}}}^{-\frac{1}{2}}$, respectively.  Note that $\mathbf\Gamma$ and $\mathbf\Omega$ are normalized matrices with unit diagonals.

\subsection{Screening Procedure}
\label{subsec:corrScreenProc}

The goal of correlation hub screening is to identify highly correlated components of the random vector $\bZ$ from its sample realizations.
Assume that $m$ samples $\bz_{1}, \ldots, \bz_{m} \in \mathbb{R}^{p}$ of 
$\bZ$ are available.  To simplify the development of the theory, the samples are assumed to be independent and identically distributed (i.i.d.) although the theory also applies to dependent samples. 

We compute sample correlation and partial correlation matrices from the samples $\bz_{1}, \ldots, \bz_{m}$ as surrogates for the unknown population correlation matrices $\mathbf\Gamma$ and $\mathbf\Omega$ in Section \ref{subsec:corrScreenModel}.  First define the $p \times p$ sample covariance matrix $\bS$ 
as $ \bS=\frac{1}{m-1} \sum_{i=1}^m
(\bz_{i}-\ol{\bz})(\bz_{i}-\ol{\bz})^H, \label{sampcov} $
where 
$\ol{\bz}$ is the sample mean, the average of $\bz_{1}, \ldots, \bz_{m}$. 
The 
sample correlation and sample partial correlation matrices are then defined as $\bR = \bD_{\bS}^{-\frac{1}{2}} \bS \bD_{\bS}^{-\frac{1}{2}}$ and $\bP = \bD_{\bR^{\dagger}}^{-\frac{1}{2}} \bR^{\dagger} \bD_{\bR^{\dagger}}^{-\frac{1}{2}}$, respectively, where $\bR^{\dagger}$ is the Moore-Penrose pseudo-inverse of $\bR$.

Correlation hubs are screened by applying thresholds to the sample (partial) correlation matrix. 
A 
variable $Z_i$ is declared a hub screening discovery at degree level $\delta \in \{1,2,\ldots\}$ and threshold level $\rho \in [0,1]$ if
\ben
|\{j: j \neq i, |{\psi}_{ij}|\geq \rho \}| \geq \delta,
\een
where $\mathbf{\Psi} = \bR$ for 
correlation screening and $\mathbf{\Psi} = \bP$ for 
partial correlation screening. We denote by $N_{\delta,\rho} \in \{0,\ldots,p\}$ the total number of hub screening discoveries at levels $\delta,\rho$.

Correlation hub screening can also be interpreted in terms of the \emph{(partial) correlation graph} $\mathcal G_{\rho}({\mathbf \Psi})$, depicted in Fig.~\ref{fig:bipartite_graph} and defined as follows. The vertices of $\mathcal G_{\rho}({\mathbf \Psi})$ are $v_1,\cdots,v_p$ which correspond to $Z_1,\cdots,Z_p$, respectively. For $1 \leq i,j \leq p$, $v_i$ and $v_j$ are connected by an edge in $\mathcal G_{\rho}({\mathbf \Psi})$ if the magnitude of the sample (partial) correlation coefficient between $Z_i$ and $Z_j$ is at least $\rho$. 
A vertex of $\mathcal G_{\rho}({\mathbf \Psi})$ is called a $\delta$-hub if its degree, the number of incident edges, is at least $\delta$. 
Then the number of discoveries $N_{\delta,\rho}$ defined earlier 
is the number of $\delta$-hubs in the graph $\mathcal G_{\rho}({\mathbf \Psi})$.

\begin{figure}
\begin{center}

\begin{tikzpicture}

\node at (1.4,0.5) [matched] {};
\node at (0.75,1.28) [matched] {};
\node at (-0.75,1.28) [matched] {};
\node at (-1.4,0.5) [matched] {};
\node at (-1.1,-1) [matched] {};
\node at (0.5,-1.4) [matched] {};

\node at (1.75,0.5) {$v_3$};
\node at (0.75,1.7) {$v_2$};
\node at (-0.75,1.7) {$v_1$};
\node at (-1.75,0.5) {$v_p$};
\node at (-1.1,-1.4) {$v_j$};
\node at (0.5,-1.8) {$v_i$};
\draw (1.4,0.5) -- (0.8,0.4);
\draw (1.4,0.5) -- (0.85,0.15);
\draw (1.4,0.5) -- (0.9,-0.1);
\draw (-1.4,0.5) -- (-0.8,0.4);
\draw (-1.4,0.5) -- (-0.85,0.15);
\draw (-1.4,0.5) -- (-0.9,-0.1);
\draw (0.75,1.28) -- (0.65,0.50);
\draw (0.75,1.28) -- (0.4,0.75);
\draw (0.75,1.28) -- (0.15,0.95);
\draw (-0.75,1.28) -- (-0.65,0.50);
\draw (-0.75,1.28) -- (-0.4,0.75);
\draw (-0.75,1.28) -- (-0.15,0.95);
\draw (0.5,-1.4) -- (0.65,-0.6);
\draw (0.5,-1.4) -- (0.40,-0.65);
\draw (0.5,-1.4) -- (0.15,-0.75);
\draw (-1.1,-1) -- (-1,-0.30);
\draw (-1.1,-1) -- (-.75,-0.55);
\draw (-1.1,-1) -- (-0.5,-0.80);
\draw [thick] (0.5,-1.4) -- (-1.1,-1);
\draw [dash pattern= on 1pt off 10pt on 1pt off 10pt](-1.4,0.5) arc (160:221:1.5 cm);/
\draw [dash pattern=on 1pt off 12pt on 1pt off 10pt] (-1.1,-1) arc (221:292:1.5 cm);/
\draw [dash pattern=on 1pt off 14pt on 1pt off 14pt](0.5,-1.4) arc (292:380:1.5 cm);/
\end{tikzpicture}
\caption{Complex-valued (partial) correlation hub screening thresholds the sample correlation or partial correlation matrix, denoted generically by the matrix $\mathbf \Psi$, to find variables $Z_i$ that are highly correlated with other variables. This is equivalent to finding hubs in a graph ${\mathcal G}_{\rho}({\mathbf \Psi})$ with $p$ vertices $v_1,\cdots,v_p$. For $1 \leq i,j \leq p,$ $v_i$ is connected to $v_j$ in ${\mathcal G}_{\rho}({\mathbf \Psi})$ if
 $|{\psi}_{ij}| \geq \rho$.} \label{fig:bipartite_graph}
\end{center}
\end{figure}
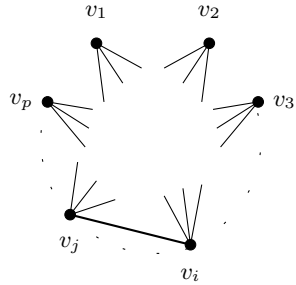

\subsection{U-score Representation of Correlation Matrices}
\label{subsec:corrScreenUScore}

Our theory for complex-valued correlation screening is based on the U-score representation of the sample correlation and partial correlation matrices. Similarly to the real case \cite{hero2012hub}, it can be shown  that there exists an $(m-1) \times p$ complex-valued matrix $\mathbb U_{\bR}$ with unit-norm columns $\bu_{\bR}^{(i)} \in \mathbb{C}^{m-1}$ such that the following representation holds:
\begin{equation} \label{eq:UR}
\bR= \mathbb U_{\bR}^H \mathbb U_{\bR}.
\end{equation}
Similar to Lemma 1 in \cite{hero2012hub} it is straightforward to show that:
\ben
\bR^{\dagger} = \mathbb U_{\bR}^H (\mathbb U_{\bR} \mathbb U_{\bR}^H)^{-2} \mathbb U_{\bR}.
\een
Hence by defining $\mathbb U_{\bP} = (\mathbb U_{\bR} \mathbb U_{\bR}^H)^{-1} \mathbb U_{\bR} \bD_{\mathbb U_{\bR}^H (\mathbb U_{\bR} \mathbb U_{\bR}^H)^{-2} \mathbb U_{\bR}} ^{-\frac{1}{2}}$
we have the representation:
\begin{equation} \label{eq:UP}
\bP = \mathbb U_{\bP}^H \mathbb U_{\bP},
\end{equation}
where the $(m-1) \times p$ matrix $\mathbb U_{\bP}$ has unit-norm columns $\bu_{\bP}^{(i)} \in \mathbb{C}^{m-1}$.

\subsection{Properties of U-scores}
\label{Props}

The U-score factorizations in \eqref{eq:UR} and \eqref{eq:UP} show that sample (partial) correlation matrices can be represented in terms of unit vectors in $\mathbb C^{m-1}$.  This subsection presents  definitions and properties related to 
U-scores that will be used in Section \ref{sec:Theoretical}.


We denote the unit spheres in $\mathbb R^{m-1}$ and $\mathbb C^{m-1}$ as $S_{m-1}$ and $T_{m-1}$, respectively. The surface areas of $S_{m-1}$ and $T_{m-1}$ are denoted as $a_{m-1}$ and $b_{m-1}$ respectively.
Define the interleaving function $h: \mathbb R^{2m-2} \rightarrow \mathbb C^{m-1}$ as below:
\be
&&h([x_1,x_2,\cdots,x_{2m-2}]^T) = \nonumber \\
&&[x_1+x_2\sqrt{-1},x_3+x_4\sqrt{-1},\cdots,x_{2m-3}+x_{2m-2}\sqrt{-1}]^T \nonumber.
\ee
Note that $h(.)$ is a one-to-one and onto function and it maps $S_{2m-2}$ to $T_{m-1}$.

For a fixed vector $\bu \in T_{m-1}$ and a threshold $0 \leq \rho \leq 1$ define the spherical cap in $T_{m-1}$:
\ben
\label{eq:arho}
A_{\rho}(\bu) = \{\bfy : \bfy \in T_{m-1}, |\bfy^H \bu| \geq \rho \}.
\een
Also define $P_0$ as the probability that a random point $\bY$ that is uniformly distributed on $T_{m-1}$ falls into $A_{\rho}(\bu)$. 
Below we give a simple expression for $P_0$ as a function of $\rho$ and $m$.

\begin{lem}
\label{Lemma:P0}
Let $\bY$ 
be an $(m-1)$-dimensional complex-valued random vector that is uniformly distributed over $T_{m-1}$. We have $P_0 = \mathbb{P} \left( \bY \in A_\rho(\bu) \right) = (1-\rho^2)^{m-2}$.
\end{lem}
\begin{proof}
Without loss of generality we assume $\bu = [1, 0, \cdots, 0]^T$. We have:
\ben
P_0 = \mathbb{P} ( |Y_1| \geq \rho ) = \mathbb{P} ( \Re(Y_{1})^2 + \Im(Y_{1})^2 \geq \rho^2 ).
\een
Since $\bY$ is uniform on $T_{m-1}$, we can write $\bY = \bX / \| \bX \|_2$, in which $\bX$ is complex-valued random vector whose entries are i.i.d. complex-valued Gaussian variables with mean $0$ and variance $1$. Thus:
\ben
P_0 &=& \mathbb{P} \left( \left(\Re(X_{1})^2 + \Im(X_{1}^2) \right) /\| \bX\|_2^2 \geq \rho^2 \right) \nonumber \\
&=& \mathbb{P} \left( (1-\rho^2) \left(\Re(X_{1})^2 + \Im(X_{1})^2 \right) \geq \rho^2 \sum_{k=2}^{m-1} \Re(X_{k})^2 + \Im(X_{k})^2  \right).
\een
Define $V_1 = \Re(X_{1})^2 + \Im(X_{1})^2$ and $V_2 = \sum_{k=2}^{m-1} \Re(X_{k})^2 + \Im(X_{k})^2$. $V_1$ and $V_2$ are independent and have chi-squared distributions with $2$ and $2(m-2)$ degrees of freedom, respectively \cite{simon2007probability}. Therefore,
\ben
P_0 &=& \int_{0}^{\infty} \int_{\rho^2 v_2 /(1-\rho^2)}^{\infty} \chi_2^2(v_1) \chi_{2(m-2)}^2(v_2) dv_1 dv_2 \nonumber \\&=& \int_0^{\infty} \chi_{2(m-2)}^2(v_2) \int_{\rho^2 v_2 /(1-\rho^2)}^{\infty} \frac{1}{2}e^{-v_1/2} dv_1 dv_2  \nonumber \\ &=&
\int_0^{\infty} \frac{1}{2^{m-2}\Gamma(m-2)}v_2^{m-3} e^{-v_2/2} e^{-\frac{\rho^2}{2(1-\rho^2)}v_2} dv_2 \nonumber  \\&=& \frac{1}{\Gamma(m-2)} (1-\rho^2)^{m-2} \int_0^{\infty} x^{m-3} e^{-x} dx \nonumber  \\&=& \frac{1}{\Gamma(m-2)} (1-\rho^2)^{m-2} \Gamma(m-2) = (1-\rho^2)^{m-2},
\een
in which we have made a change of variable $x=\frac{v_2}{2(1-\rho^2)}$.
\qed
\end{proof}

Under the assumption that the joint pdf of $\bZ$ exists, the $p$ columns of the U-score matrix have joint pdf $f_{\bU_1,\ldots,\bU_{p}}(\bu_1,\ldots,\bu_{p})$ on
$T_{m-1}^{p}={\large \times}_{i=1}^{p} T_{m-1}$. The following $(\delta+1)$-fold average of the joint pdf will play a significant role in Section \ref{sec:Theoretical}. This $(\delta+1)$-fold average is defined as:
\ben &&\ol{f_{\bU_{\ast1}, \ldots, \bU_{\ast{\delta+1}}}}
(\bu_1, \ldots, \bu_{\delta+1})  =\frac{1}{(2 \pi)^{\delta+1} p \binom{p-1}{\delta}} \times \label{eq:avgfubivpersistent}
\label{eq:olfdef1}\\
&&~~~~ \sum_{1 \leq i_1 < \cdots < i_{\delta} \leq p, i_{\delta+1} \notin \{ i_1,\cdots,i_{\delta} \} } \int_{0}^{2\pi} \nonumber \int_{0}^{2\pi}  \cdots \int_{0}^{2\pi} \nonumber \\ &&f_{\bU_{i_1},\ldots,
\bU_{i_{\delta}},\bU_{i_{\delta+1}}} (e^{\sqrt{-1}\theta_1} \bu_1, \ldots, e^{\sqrt{-1}\theta_{\delta}}
\bu_{\delta}, e^{\sqrt{-1}\theta} \bu_{\delta+1}) ~ d\theta_1 \cdots d\theta_{\delta} ~ d\theta.
\nonumber
\een
Also for a joint pdf $f_{\bU_1,\ldots,\bU_{\delta+1}}(\bu_1,\ldots,\bu_{\delta+1})$ on $T_{m-1}^{\delta+1}$ define \be &&J(f_{\bU_1,\ldots,\bU_{\delta+1}}) = \nonumber a_{2m-2}^{\delta} \int_{S_{2m-2}} f_{\bU_1,\ldots,\bU_{\delta+1}}(h(\bu), \ldots,h(\bu)) d\bu.
\label{eq:Jdef}
\ee
Note that $J(f_{\bU_1,\ldots,\bU_{\delta+1}})$ is proportional to the integral of $f_{\bU_1,\ldots,\bU_{\delta+1}}$ over the manifold $\bu_1=\ldots=\bu_{\delta+1}$. The quantity $J(\ol{f_{\bU_{\ast1}, \ldots, \bU_{\ast(\delta+1)}}})$ 
is key in determining the asymptotic average number of hubs in a complex-valued correlation network. This will be described in more detail in Sec. \ref{sec:Theoretical}.

Let $\vec{i}=(i_0,i_1,\ldots,i_\delta)$ be a set of distinct indices, i.e., 
$1 \leq i_0 \leq p, 1 \leq i_1 < \ldots < i_\delta \leq p$ and $i_1,\ldots,i_\delta \neq i_0$. For a U-score matrix $\mU$ define the
dependency coefficient between the columns $\bU_{\vec{i}}=\{\bU_{i_0},\bU_{i_1}, \ldots,
\bU_{i_{\delta}}\}$ and
their complementary $k$-NN ($k$-nearest neighbor) set $A_k(\vec{i})$ defined in \eqref{eq:CompKNN} 
and Fig. \ref{fig:KNNGraph1} as
\ben
\Delta_{p,m,k,\delta}(\vec{i})= \left \|(f_{\bU_{\vec{i}}|\bU_{A_k(\vec{i})}}-f_{\bU_{\vec{i}}})/f_{\bU_{\vec{i}}}
\right\|_{\infty},
\label{eq:deltapij}
\een
where $\lVert \cdot \rVert_{\infty}$ denotes the supremum norm. The average of these  coefficients is defined as:
\be
\|\Delta_{p,m,k,\delta}\|_1= \frac{1}{p{p-1 \choose \delta} }
\sum_{i_0=1}^{p}\sum_{\stackrel{i_1 ,\ldots, i_\delta \neq i_0} {1 \leq i_1< \ldots < i_\delta \leq p}}
\Delta_{p,m,k,\delta}(\vec{i}).
\label{eq:Deltapdefavg}
\ee%

\subsection{Number of Hub Discoveries in the High-Dimensional Limit}
\label{sec:Theoretical}

We now present the main theoretical result on complex-valued correlation screening. 
The following theorem gives asymptotic expressions for the mean number of $\delta$-hubs and the probability of discovery of at least one $\delta$-hub in the graph $\mathcal G_{\rho}({\mathbf \Psi})$. It also gives bounds on the rates of convergence to these approximations as the dimension $p$ increases and $\rho \rightarrow 1$. We use $\mathbb U = [\bU_{1},\cdots,\bU_p]$ as a generic notation for the U-score representation of the sample (partial) correlation matrix. The asymptotic expression for the mean $\mathbb E[N_{\delta,\rho}]$ is denoted by $\Lambda$ and is given by:
\be
\Lambda = p{p-1 \choose \delta} P_0^{\delta} J(\ol{f_{\bU_{\ast1}, \ldots, \bU_{\ast(\delta+1)}}}).
\label{eq:Lambdadef}
\ee
Define $\eta_{p,\delta}$ as:
\be 
\eta_{p,\delta} = p^{1/\delta}(p-1) P_0 = p^{1/\delta}(p-1)(1-\rho^2)^{(m-2)}, \label{eq:etadef}
\ee
where the last equation is due to Lemma \ref{Lemma:P0}. The parameter $k$ below represents an upper bound on the true hub degree, i.e. the number of non-zero entries in any row of the population covariance matrix $\mathbf \Sigma$. Also let $\varphi(\delta)$ be the function that takes values $\varphi(\delta)=2$ for $\delta=1$ and $\varphi(\delta)=1$ for $\delta > 1$.
\begin{thm}
Let $\mU=[\bU_1, \ldots, \bU_p]$ be a $(m-1)\times p$ random matrix
with $\bU_i \in T_{m-1}$ where $m>2$. Let $\delta \geq 1$ be a fixed integer. Assume the joint pdf of any subset of the $\bU_i$'s is bounded and differentiable. Then, with $\Lambda$ defined in (\ref{eq:Lambdadef}),
\be
\left|\mathbb E[N_{\delta,\rho}] - \Lambda\right| \leq
O\left( \eta_{p,\delta}^{\delta}
\max\left\{\eta_{p,\delta} p^{-1/\delta} ,(1-\rho)^{1/2}\right\}\right).
\label{eq:ENlim}
\ee
Furthermore, let $N^*_{\delta,\rho}$
be a Poisson distributed random variable with rate
$\mathbb E[N^*_{\delta,\rho}]=\Lambda/\varphi(\delta)$.
If $(p-1)P_0\leq 1$, then
\be
&&\left|\mathbb{P}(N_{\delta,\rho}>0)-\mathbb{P}(N^*_{\delta,\rho} >0)\right|\leq \nonumber \\ &&\left\{\begin{array}{cc}
O\left(\eta_{p,\delta}^{\delta}
\max\left\{\eta_{p,\delta}^{\delta}
\left(k/p\right)^{\delta+1},
Q_{p,k,\delta}, \|\Delta_{p,m,k,\delta}\|_1,p^{-1/\delta},(1-\rho)^{1/2}\right\}\right), & \delta>1 \\
O\left(\eta_{p,1}
\max\left\{\eta_{p,1}
\left(k/p\right)^{2}, \|\Delta_{p,m,k,1}\|_1, p^{-1},(1-\rho)^{1/2}\right\}\right), & \delta=1
\end{array}\right. ,
\nonumber
\\
\label{eq:Pvlim} \ee with $Q_{p,k,\delta}=\eta_{p,\delta}
\left(k/p^{1/\delta}\right)^{\delta+1}$ and
$\|\Delta_{p,m,k,\delta}\|_1$ defined in \eqref{eq:Deltapdefavg}.
\label{prop:parcor}
\end{thm}

\begin{proof}
The proof is similar to the proof of proposition 1 in \cite{hero2012hub}. First we prove \eqref{eq:ENlim}. Let $\phi_i=I(d_i\geq \delta)$ be the indicator of the event that $d_i\geq \delta$, in which $d_i$ represents the degree of the vertex $v_i$ in the graph ${\mathcal G}_{\rho}({\mathbf \Psi})$. We have  $N_{\delta,\rho} = \sum_{i=1}^p \phi_i$. 
With $\phi_{ij}$ being the indicator of the presence of an edge in ${\mathcal G}_{\rho}({\mathbf \Psi})$ between vertices $v_i$ and $v_j$ we have the relation:
\be \phi_i&=&\sum_{l=\delta}^{p-1}
\sum_{\vec{k}\in\breve{\calC}_i(p-1,l)} \prod_{j=1}^l \phi_{ik_j}
\prod_{q=l+1}^{p-1} (1-\phi_{i k_q}) \label{eq:phirep} \ee where
we have defined the index vector $\vec{k}=(k_1, \ldots, k_{p-1})$
and the set
$$\breve{\calC}_i(p-1,l)=$$ $$\{\vec{k}: k_1< \ldots< k_l, k_{l+1} < \ldots < k_{p-1}
\; k_j \in \{1, \ldots, p\}-\{i\}, k_j\neq k_{j'}\}.$$ The inner
summation in (\ref{eq:phirep}) simply sums over the set of
distinct indices not equal to $i$ that index all ${p-1 \choose l}$
different types of products of the form: $\prod_{j=1}^l \phi_{ik_j}
\prod_{q=l+1}^{p-1} (1-\phi_{i k_q})$. Subtracting $
\sum_{\vec{k}\in\breve{\calC}_i(p-1,\delta)}  \prod_{j=1}^\delta
\phi_{ik_j}$ from both sides of (\ref{eq:phirep})
\be &&\phi_i-\sum_{\vec{k} \in \breve{\calC}_i(p-1,\delta)}
\prod_{j=1}^\delta \phi_{ik_j} \nonumber \\ &&\hspace{0.2in}=
\sum_{l=\delta+1}^{p-1} \sum_{\vec{k}\in\breve{\calC}_i(p-1,l)}
\prod_{j=1}^l \phi_{ik_j}\prod_{q=l+1}^{p-1} (1-\phi_{i k_q}) \nonumber
\\
&&\hspace{0.3in}+\sum_{\vec{k}\in\breve{\calC}_i(p-1,l)}
\sum_{q=\delta+1}^{p-1}(-1)^{q-\delta} \nonumber \\ && \sum_{k'_{\delta+1} <
\ldots <k'_q, \{k'_{\delta+1},...,k'_q \} \subset \{
k_{\delta+1},...,k_{p-1}\} } \prod_{j=1}^l \phi_{ik_j}
\prod_{s=\delta+1}^q \phi_{ik'_s} \label{eq:phirep2} \ee in which we have used the expansion
$$\prod_{q=\delta+1}^{p-1} (1-\phi_{i k_q})=
1+ \sum_{q=\delta+1}^{p-1}(-1)^{q-\delta} \sum_{k'_{\delta+1} <
\ldots <k'_{q}, \{k'_{\delta+1},...,k'_q \} \subset \{
k_{\delta+1},...,k_{p-1}\}} \prod_{s=\delta+1}^q \phi_{ik'_s}.$$

The following simple asymptotic representation will be useful in
the sequel. For any $i_1, \ldots, i_k \in \{1, \ldots , p\}$, $i_1
\neq \cdots \neq i_k \neq i$, $k\in \{1, \ldots, p-1\}$,
\be
\mathbb E\left[\prod_{j=1}^k \phi_{i i_j} \right] &=&\int_{S_{2m-2}}
 \int_{h^{-1}(A_{\rho}(\bfv))}  \cdots \int_{h^{-1}(A_{\rho}(\bfv))}
 \nonumber \\ && f_{\bfU_{i_1},\ldots, \bfU_{i_k},\bfU_i}(h(\bfv_1),
\cdots,h(\bfv_k),h(\bfv)) ~ d\bfv_{1} \cdots d\bfv_{k} ~ d\bfv
\label{eq:thetamoms} \nonumber \\
&\leq& P_0^k a_{2m-2}^k M_{k|1} \label{eq:thetamomsineq}
\label{eq:old714}
\ee
where $P_0, A_{\rho}(\bu)$ and the function $h(.)$ are defined in Sec. \ref{Props}.  Moreover \ben M_{k|1} &=&\max_{i_1 \neq \cdots \neq
i_{k+1}}\left\| f_{\bfU_{i_1},\ldots,
\bfU_{i_k}|\bfU_{i_{k+1}}}\right\|_{\infty} \label{eq:Mk1def}.
\een  
The following simple generalization of \eqref{eq:old714} to arbitrary
product indices $\phi_{ij}$ will also be needed \be
\mathbb E\left[\prod_{l=1}^q \phi_{i_lj_l}\right] \leq  P_0^q a_{2m-2}^q
M_{|Q|}, \label{eq:thetamoms2} \ee where
$Q=$unique$(\{i_l,j_l\}_{l=1}^q)$ is the set of unique indices
among the distinct pairs $\{(i_l,j_l)\}_{l=1}^q$ and $M_{|Q|}$ is
a bound on the joint pdf of $\bU_{Q}$.

Define the random variable \ben \theta_i = {p-1 \choose \delta}^{-1}
\sum_{\vec{k} \in \breve{\calC}_i(p-1,\delta)} \prod_{j=1}^\delta
\phi_{ik_j}. \label{eq:thetadefnew} \een We show below that for
sufficiently large $p$ \be \left|\mathbb E[\phi_i]-{p-1 \choose \delta}
\mathbb E[\theta_i]\right| &\leq& \gamma_{p,\delta} ((p-1)
P_0)^{\delta+1}, \label{eq:proof1} \ee where
$\gamma_{p,\delta}=\max_{\delta+1\leq l<p} \{a_{2m-2}^l M_{l|1}\}
\left(e-\sum_{l=0}^\delta \frac{1}{l!}\right)
\left(1+(\delta!)^{-1}\right)$ and $M_{l|1}$ is a least upper
bound on any $l$-dimensional joint pdf of the variables
$\{\bU_i\}_{j\neq i}^p$ conditioned on $\bU_i$.

To show inequality (\ref{eq:proof1}) take expectations of
(\ref{eq:phirep2})
and apply the bound (\ref{eq:old714}) 
to obtain \be
&&\left|\mathbb E[\phi_i]-{p-1 \choose \delta} \mathbb E[\theta_i]\right| \leq \nonumber \\
&& \left|\sum_{l=\delta+1}^{p-1} {p-1 \choose
l}  P_0^l a_{2m-2}^lM_{l|1} +{p-1 \choose \delta}
\sum_{l=1}^{p-1-\delta} {p-1-\delta \choose l} P_0^{\delta+l}
a_{2m-2}^{\delta+l}
M_{\delta+l|1} \right| \nonumber \\
&& \leq A(1+(\delta!)^{-1}), \label{eq:proof2} \ee
where
$$A=\sum_{l=\delta+1}^{p-1} {p-1 \choose l} ((p-1)P_0)^l a_{2m-2}^lM_{l|1}.$$
The line (\ref{eq:proof2}) follows from the identity ${p-1-\delta
\choose l}{p-1 \choose \delta} ={p-1 \choose l+\delta} {l+\delta
\choose l}$ and a change of index in the second summation on the
previous line. Since $(p-1)P_0<1$
\ben |A|
&\leq& \max_{\delta+1\leq l<p} \{a_{2m-2}^l M_{l|1}\}
\sum_{l=\delta+1}^{p-1} {p-1 \choose l}
((p-1)P_0)^l \\
&\leq& \max_{\delta+1\leq l<p} \{a_{2m-2}^l
M_{l|1}\}\left(e-\sum_{l=0}^\delta \frac{1}{l!}\right)
((p-1)P_0)^{\delta+1}. \een

Application of the mean value theorem to the integral
representation (\ref{eq:thetamoms}) yields
\be \left|\mathbb E[\theta_i]- P_0^\delta J(\ol{f_{\bU_{\ast1-i}, \ldots,
\bU_{\ast\delta-i},\bU_i}}) \right| &\leq&
\tilde{\gamma}_{p,\delta} ((p-1)P_0)^\delta r, \label{eq:proof3}
 \ee
where 
\ben \ol{f_{\bU_{\ast1-i}, \ldots,
\bU_{\ast\delta-i},\bU_i}}
(\bu_1, && \ldots, \bu_{\delta+1})  = \\ \frac{1}{(2 \pi)^{\delta} \binom{p-1}{\delta}} && 
 \sum_{\stackrel{1 \leq i_1 < \cdots < i_{\delta} \leq p}{i \notin \{i_1,\cdots,i_{\delta} \} }} \nonumber \int_{0}^{2\pi}  \cdots \int_{0}^{2\pi} \nonumber \\ f_{\bU_{i_1},\ldots,
\bU_{i_{\delta}},\bU_{i}} (e^{\sqrt{-1}\theta_1} \bu_1, &&\ldots, e^{\sqrt{-1}\theta_{\delta}}
\bu_{\delta}, \bu_{\delta+1}) ~  d\theta_1 \cdots d\theta_{\delta},
\nonumber
\een
$r=\sqrt{2(1-\rho)}$,
$\tilde{\gamma}_{p,\delta}= 2a_{2m-2}^{\delta+1}
\dot{M}_{\delta+1|1}/\delta!$ and $\dot{M}_{\delta+1|1}$ is a
bound on the norm of the gradient
$$\nabla_{\bu_{i_1}, \ldots, \bu_{i_\delta}}\ol{f_{\bU_{\ast1-i}, \ldots, \bU_{\ast\delta-i}|\bU_i}(\bu_{i_1}, \ldots, \bu_{i_\delta}|\bu_i)}.$$
Combining (\ref{eq:proof1})-(\ref{eq:proof3}) and the relation
$r=O((1-\rho)^{1/2})$, \ben && \left|\mathbb E[\phi_{i}]- {p-1 \choose \delta}
P_0^{\delta} J(\ol{f_{\bU_{\ast1}, \ldots,
            \bU_{\ast(\delta+1)}}})\right| \nonumber \\ &\leq& O\left(((p-1)P_0)^{\delta} \max\left\{ (p-1)P_0, (1-\rho)^{1/2}\right\} \right).
\label{eq:proof3pp}  \een

Summing over $i$ and recalling  the
definitions (\ref{eq:Lambdadef})  and  (\ref{eq:etadef}) of
$\Lambda$ and $\eta_{p,\delta}$, \ben
\left|\mathbb E[N_{\delta,\rho}]-\Lambda \right| &\leq&
O\left(p((p-1)P_0)^{\delta} \max\left\{ (p-1)P_0,
(1-\rho)^{1/2}\right\} \right) \nonumber
\\
&=&O\left(\eta_{p,\delta}^{\delta} \max\left\{\eta_{p,\delta}
p^{-1/\delta}, (1-\rho)^{1/2}\right\} \right). \label{eq:proof3p}
\een
This establishes the bound (\ref{eq:ENlim}).

Next we prove the bound (\ref{eq:Pvlim}) by using the Chen-Stein method
\cite{arratia1990poisson}.
Define:
\be \tilde{N}_{\delta,\rho}= \frac{1}{\varphi(\delta)} \sum_{i_0=1}^{p} \sum_{1\leq i_1<\ldots <
i_{\delta}\leq p}\prod_{j=1}^{\delta}\phi_{i_0i_j}
\label{eq:rep1}, \ee
Where the second sum is over the indices $1\leq i_1<\ldots <
i_{\delta}\leq p$ such that $i_{j} \neq i_0, 1 \leq j \leq \delta$. For $\vec{i}\defined (i_0,i_1, \ldots, i_{\delta})$ define the
index set $\Beta_{\vec{i}}=\Beta_{i_0,i_1,\ldots,
i_{\delta}}=\{(j_0,j_1,\ldots, j_{\delta}): j_l\in \mathcal
N_k(i_l)\cup\{i_l\}, l=0, \ldots, \delta\}\cap {\mathcal C}^{<}$
where ${\mathcal C}^{<}=\{ (j_0, \ldots, j_\delta): 1 \leq j_0 \leq p, 1\leq j_1<
\cdots < j_\delta \leq p, j_{l} \neq j_0, 1 \leq l \leq \delta \}$.  These index the distinct sets of
points $\bU_{\vec{i}}=\{\bU_{i_0},\bU_{i_1}, \ldots,
\bU_{i_{\delta}}\}$ and their respective
$k$-NN's.  Note that $|\Beta_{\vec{i}}|\leq k^{\delta+1}$. 
Identifying $\tilde{N}_{\delta,\rho}= \sum_{\vec{i}\in {\mathcal
C}^{<}} \prod_{l=1}^{\delta}\phi_{i_0i_l}$ and $N_{\delta,\rho}^*$
a Poisson distributed random variable with rate
$\mathbb E[\tilde{N}_{\delta,\rho}]$, the Chen-Stein bound \cite[Theorem 
1]{arratia1990poisson} is
\be 2\max_A |\mathbb{P}(\tilde{N}_{\delta,\rho} \in
A)-\mathbb{P}(N_{\delta,\rho}^*\in A)| \leq b_1+b_2+b_3,
\label{eq:chenstein} \ee where
$$b_1 = \sum_{\vec{i}\in {\mathcal C}^{<}}\sum_{\vec{j}\in \Beta_{\vec{i}}}
\mathbb E\left[\prod_{l=1}^{\delta}\phi_{i_0i_l}\right]\mathbb E\left[\prod_{q=1}^{\delta}\phi_{j_0j_q}\right],
$$
$$b_2 =
\sum_{\vec{i}\in {\mathcal C}^{<}}\sum_{\vec{j}\in
\Beta_{\vec{i}-\{\vec{i}\}}}
\mathbb E\left[\prod_{l=1}^{\delta}\phi_{i_0i_l}\prod_{q=1}^{\delta}\phi_{j_0j_q}\right],$$
and, for $p_{\vec{i}}=\mathbb E[\prod_{l=1}^{\delta}\phi_{i_0i_l}]$,
$$
b_3 = \sum_{\vec{i}\in {\mathcal C}^{<}} \mathbb E\left[
\mathbb E\left[\left.\prod_{l=1}^{\delta}\phi_{i_0i_l}-p_{\vec{i}}\right|\phi_{\vec{j}}:
\vec{j} \not\in \Beta_{\vec{i}}\right]\right].$$

Over the range of indices in the sum of $b_1$
$\mathbb E[\prod_{l=1}^{\delta}\phi_{ii_l}]$ is of order
$O(P_0^{\delta})$, by (\ref{eq:thetamoms2}), and therefore
$$b_1\leq O\left(p^{\delta+1} k^{\delta+1} P_0^{2\delta} \right)=
O\left(\eta_{p,\delta}^{2\delta}(k/p)^{\delta+1}\right),$$ which
follows from definition (\ref{eq:etadef}).
More care is needed to bound $b_2$ due to the repetition of
characteristic functions $\phi_{ij}$. Since $\vec{i} \neq
\vec{j}$,
$\mathbb E[\prod_{l=1}^{\delta}\phi_{i_0i_l}\prod_{q=1}^{\delta}\phi_{j_0j_q}]$
is a multiplication of at least $\delta+1$ different
characteristic functions, hence by (\ref{eq:thetamoms2}),
$$\mathbb E[\prod_{l=1}^{\delta}\phi_{i_0i_l}\prod_{q=1}^{\delta}\phi_{j_0j_q}] = O\left(P_0^{\delta+1}\right).$$
Therefore, we conclude that \ben b_2&\leq& O\left(p^{\delta+1}
k^{\delta+1} P_0^{\delta+1} \right). \een


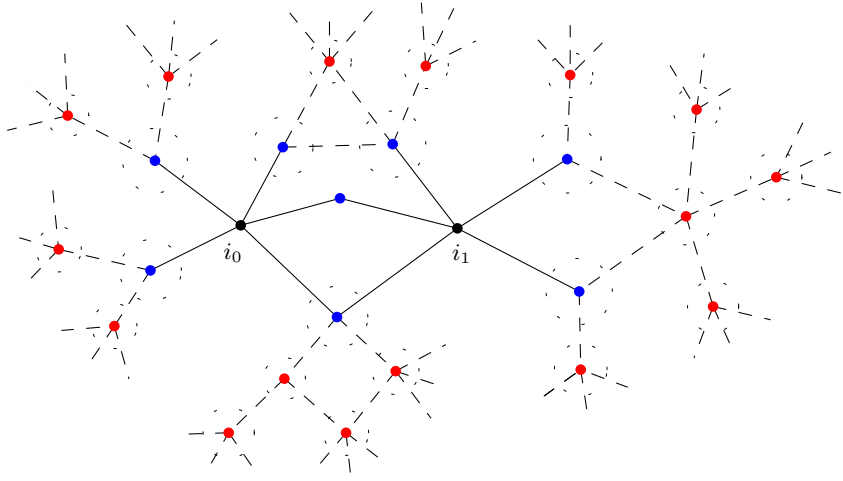
\begin{figure}
\begin{center}
\begin{tikzpicture}[line cap=round,line join=round]
\draw [dash pattern=on 5pt off 5pt] (7.08,4.18)-- (6.94,2.76);
\draw [dash pattern=on 5pt off 5pt] (6.94,2.76)-- (7.3,1.56);
\draw [dash pattern=on 5pt off 5pt] (6.94,2.76)-- (8.14,3.28);
\draw (3.9,2.6)-- (5.36,3.52);
\draw [dash pattern=on 5pt off 5pt] (5.36,3.52)-- (6.94,2.76);
\draw [dash pattern=on 5pt off 5pt] (6.94,2.76)-- (5.52,1.76);
\draw (3.9,2.6)-- (5.52,1.76);
\draw [dash pattern=on 5pt off 5pt] (5.52,1.76)-- (5.54,0.72);
\draw [dash pattern=on 5pt off 5pt] (5.4,4.64)-- (5.36,3.52);
\draw (3.04,3.72)-- (3.9,2.6);
\draw [dash pattern=on 5pt off 5pt] (3.04,3.72)-- (3.48,4.76);
\draw [dash pattern=on 5pt off 5pt] (3.04,3.72)-- (2.2,4.82);
\draw [dash pattern=on 5pt off 5pt] (2.2,4.82)-- (1.58,3.68);
\draw [dash pattern=on 5pt off 5pt] (1.58,3.68)-- (3.04,3.72);
\draw (2.34,3)-- (3.9,2.6);
\draw (1.58,3.68)-- (1.02,2.64);
\draw (1.02,2.64)-- (2.34,3);
\draw (1.02,2.64)-- (2.3,1.42);
\draw (2.3,1.42)-- (3.9,2.6);
\draw [dash pattern=on 5pt off 5pt] (2.3,1.42)-- (3.08,0.7);
\draw [dash pattern=on 5pt off 5pt] (2.3,1.42)-- (1.6,0.6);
\draw [dash pattern=on 5pt off 5pt] (1.6,0.6)-- (2.42,-0.12);
\draw [dash pattern=on 5pt off 5pt] (3.08,0.7)-- (2.42,-0.12);
\draw [dash pattern=on 5pt off 5pt] (1.6,0.6)-- (0.86,-0.12);
\draw (-0.12,3.5)-- (1.02,2.64);
\draw [dash pattern=on 5pt off 5pt] (0.06,4.62)-- (-0.12,3.5);
\draw [dash pattern=on 5pt off 5pt] (-0.12,3.5)-- (-1.28,4.1);
\draw (-0.18,2.04)-- (1.02,2.64);
\draw [dash pattern=on 5pt off 5pt] (-0.18,2.04)-- (-0.66,1.3);
\draw [dash pattern=on 5pt off 5pt] (-0.18,2.04)-- (-1.4,2.32);
\draw [dash pattern=on 5pt off 5pt] (4.96,5.12)-- (5.4,4.64);
\draw [dash pattern=on 5pt off 5pt] (5.4,5.32)-- (5.4,4.64);
\draw [dash pattern=on 5pt off 5pt] (5.84,5.08)-- (5.4,4.64);
\draw [dash pattern=on 5pt off 5pt] (6.72,4.78)-- (7.08,4.18);
\draw [dash pattern=on 5pt off 5pt] (7.08,4.18)-- (7.66,4.54);
\draw [dash pattern=on 5pt off 5pt] (7.08,4.18)-- (7.18,4.92);
\draw [dash pattern=on 5pt off 5pt] (8.14,3.28)-- (8.88,3.62);
\draw [dash pattern=on 5pt off 5pt] (8.5,4.02)-- (8.14,3.28);
\draw [dash pattern=on 5pt off 5pt] (8.14,3.28)-- (9.04,3.06);
\draw [dash pattern=on 5pt off 5pt] (7.3,1.56)-- (7.56,0.98);
\draw [dash pattern=on 5pt off 5pt] (7.3,1.56)-- (7.04,0.86);
\draw [dash pattern=on 5pt off 5pt] (7.3,1.56)-- (8.08,1.38);
\draw [dash pattern=on 5pt off 5pt] (5.04,0.36)-- (5.54,0.72);
\draw [dash pattern=on 5pt off 5pt] (5.04,0.36)-- (5.54,0.72);
\draw [dash pattern=on 5pt off 5pt] (6.18,0.48)-- (5.54,0.72);
\draw [dash pattern=on 5pt off 5pt] (5.54,0.72)-- (5.64,0.06);
\draw [dash pattern=on 5pt off 5pt] (3.08,0.7)-- (3.68,0.5);
\draw [dash pattern=on 5pt off 5pt] (3.08,0.7)-- (3.76,1.06);
\draw [dash pattern=on 5pt off 5pt] (3.08,0.7)-- (3.56,0.06);
\draw [dash pattern=on 5pt off 5pt] (2.42,-0.12)-- (2.94,-0.52);
\draw [dash pattern=on 5pt off 5pt] (2.42,-0.12)-- (2,-0.62);
\draw [dash pattern=on 5pt off 5pt] (2.42,-0.12)-- (2.48,-0.68);
\draw [dash pattern=on 5pt off 5pt] (0.86,-0.12)-- (1.18,-0.58);
\draw [dash pattern=on 5pt off 5pt] (0.86,-0.12)-- (0.22,-0.14);
\draw [dash pattern=on 5pt off 5pt] (0.62,-0.54)-- (0.86,-0.12);
\draw [dash pattern=on 5pt off 5pt] (3.48,4.76)-- (3.42,5.62);
\draw [dash pattern=on 5pt off 5pt] (3.48,4.76)-- (4.16,5.1);
\draw [dash pattern=on 5pt off 5pt] (3.86,5.54)-- (3.48,4.76);
\draw [dash pattern=on 5pt off 5pt] (1.66,5.48)-- (2.2,4.82);
\draw [dash pattern=on 5pt off 5pt] (2.2,4.82)-- (2.2,5.52);
\draw [dash pattern=on 5pt off 5pt] (2.2,4.82)-- (2.78,5.5);
\draw [dash pattern=on 5pt off 5pt] (-1.94,2.48)-- (-1.4,2.32);
\draw [dash pattern=on 5pt off 5pt] (-1.48,3.12)-- (-1.4,2.32);
\draw [dash pattern=on 5pt off 5pt] (-1.4,2.32)-- (-1.86,1.84);
\draw [dash pattern=on 5pt off 5pt] (-1.38,1.24)-- (-0.66,1.3);
\draw [dash pattern=on 5pt off 5pt] (-0.66,1.3)-- (-0.46,0.58);
\draw [dash pattern=on 5pt off 5pt] (-0.66,1.3)-- (-1.04,0.74);
\draw [dash pattern=on 5pt off 5pt] (-1.28,4.1)-- (-1.84,4.54);
\draw [dash pattern=on 5pt off 5pt] (-1.28,4.1)-- (-1.32,4.94);
\draw [dash pattern=on 5pt off 5pt] (-1.28,4.1)-- (-2.1,3.9);
\draw [dash pattern=on 5pt off 5pt] (0.06,4.62)-- (-0.56,5.38);
\draw [dash pattern=on 5pt off 5pt] (0.06,4.62)-- (0.72,5.12);
\draw [dash pattern=on 5pt off 5pt] (0.06,4.62)-- (0.14,5.38);
\draw [line width=0.2pt,dash pattern=on 1pt off 9pt] (-0.12,3.5) circle (0.43cm);
\draw [line width=0.2pt,dash pattern=on 1pt off 9pt] (-0.18,2.04) circle (0.41cm);
\draw [line width=0.2pt,dash pattern=on 1pt off 9pt] (2.3,1.42) circle (0.41cm);
\draw [line width=0.2pt,dash pattern=on 1pt off 9pt] (5.52,1.76) circle (0.44cm);
\draw [line width=0.2pt,dash pattern=on 1pt off 9pt] (5.36,3.52) circle (0.45cm);
\draw [line width=0.2pt,dash pattern=on 1pt off 9pt] (1.58,3.68) circle (0.4cm);
\draw [line width=0.2pt,dash pattern=on 1pt off 9pt] (3.04,3.72) circle (0.44cm);
\draw [line width=0.2pt,dash pattern=on 1pt off 9pt] (7.08,4.18) circle (0.31cm);
\draw [line width=0.2pt,dash pattern=on 1pt off 9pt] (8.14,3.28) circle (0.34cm);
\draw [line width=0.2pt,dash pattern=on 1pt off 9pt] (6.94,2.76) circle (0.38cm);
\draw [line width=0.2pt,dash pattern=on 1pt off 9pt] (7.3,1.56) circle (0.34cm);
\draw [line width=0.2pt,dash pattern=on 1pt off 9pt] (5.4,4.64) circle (0.3cm);
\draw [line width=0.2pt,dash pattern=on 1pt off 9pt] (3.48,4.76) circle (0.3cm);
\draw [line width=0.2pt,dash pattern=on 1pt off 9pt] (2.2,4.82) circle (0.29cm);
\draw [line width=0.2pt,dash pattern=on 1pt off 9pt] (0.06,4.62) circle (0.3cm);
\draw [line width=0.2pt,dash pattern=on 1pt off 9pt] (-1.28,4.1) circle (0.28cm);
\draw [line width=0.2pt,dash pattern=on 1pt off 9pt] (-1.4,2.32) circle (0.3cm);
\draw [line width=0.2pt,dash pattern=on 1pt off 9pt] (-0.66,1.3) circle (0.3cm);
\draw [line width=0.2pt,dash pattern=on 1pt off 9pt] (1.6,0.6) circle (0.28cm);
\draw [line width=0.2pt,dash pattern=on 1pt off 9pt] (5.54,0.72) circle (0.35cm);
\draw [line width=0.2pt,dash pattern=on 1pt off 9pt] (3.08,0.7) circle (0.31cm);
\draw [line width=0.2pt,dash pattern=on 1pt off 9pt] (2.42,-0.12) circle (0.3cm);
\draw [line width=0.2pt,dash pattern=on 1pt off 9pt] (0.86,-0.12) circle (0.34cm);
\begin{small}
\fill [color=black] (3.9,2.6) circle (2.0pt);
\draw[color=black] (3.96,2.26) node {$i_1$};
\fill [color=red] (6.94,2.76) circle (2.0pt);
\fill [color=blue] (5.36,3.52) circle (2.0pt);
\fill [color=blue] (5.52,1.76) circle (2.0pt);
\fill [color=red] (7.08,4.18) circle (2.0pt);
\fill [color=red] (7.3,1.56) circle (2.0pt);
\fill [color=red] (8.14,3.28) circle (2.0pt);
\fill [color=red] (5.54,0.72) circle (2.0pt);
\fill [color=red] (5.4,4.64) circle (2.0pt);
\fill [color=blue] (2.34,3) circle (2.0pt);
\fill [color=blue] (3.04,3.72) circle (2.0pt);
\fill [color=blue] (1.58,3.68) circle (2.0pt);
\fill [color=red] (3.48,4.76) circle (2.0pt);
\fill [color=red] (2.2,4.82) circle (2.0pt);
\fill [color=blue] (2.3,1.42) circle (2.0pt);
\fill [color=black] (1.02,2.64) circle (2.0pt);
\draw[color=black] (0.92,2.28) node {$i_0$};
\fill [color=red] (1.6,0.6) circle (2.0pt);
\fill [color=red] (3.08,0.7) circle (2.0pt);
\fill [color=red] (2.42,-0.12) circle (2.0pt);
\fill [color=red] (0.86,-0.12) circle (2.0pt);
\fill [color=blue] (-0.12,3.5) circle (2.0pt);
\fill [color=red] (0.06,4.62) circle (2.0pt);
\fill [color=red] (-1.28,4.1) circle (2.0pt);
\fill [color=blue] (-0.18,2.04) circle (2.0pt);
\fill [color=red] (-0.66,1.3) circle (2.0pt);
\fill [color=red] (-1.4,2.32) circle (2.0pt);
\fill [color=white] (4.96,5.12) circle (0.5pt);
\fill [color=white] (5.4,5.32) circle (0.5pt);
\fill [color=white] (5.84,5.08) circle (0.5pt);
\fill [color=white] (5.38,5.92) circle (0.5pt);
\fill [color=white] (6.72,4.78) circle (0.5pt);
\fill [color=white] (7.66,4.54) circle (0.5pt);
\fill [color=white] (7.18,4.82) circle (0.5pt);
\fill [color=white] (7.32,5.46) circle (0.5pt);
\fill [color=white] (8.5,4.02) circle (0.5pt);
\fill [color=white] (9.04,3.06) circle (0.5pt);
\fill [color=white] (8.88,3.62) circle (0.5pt);
\fill [color=white] (9.6,3.92) circle (0.5pt);
\fill [color=white] (7.04,0.86) circle (0.5pt);
\fill [color=white] (7.56,0.98) circle (0.5pt);
\fill [color=white] (8.08,1.38) circle (0.5pt);
\fill [color=white] (7.78,0.34) circle (0.5pt);
\fill [color=white] (5.04,0.36) circle (0.5pt);
\fill [color=white] (6.18,0.48) circle (0.5pt);
\fill [color=white] (5.64,0.06) circle (0.5pt);
\fill [color=white] (5.72,-0.52) circle (0.5pt);
\fill [color=white] (3.76,1.06) circle (0.5pt);
\fill [color=white] (3.68,0.5) circle (0.5pt);
\fill [color=white] (3.56,0.06) circle (0.5pt);
\fill [color=white] (2,-0.62) circle (0.5pt);
\fill [color=white] (2.94,-0.52) circle (0.5pt);
\fill [color=white] (2.48,-0.68) circle (0.5pt);
\fill [color=white] (0.22,-0.14) circle (0.5pt);
\fill [color=white] (0.62,-0.54) circle (0.5pt);
\fill [color=white] (1.18,-0.58) circle (0.5pt);
\fill [color=white] (0.36,-1.02) circle (0.5pt);
\fill [color=white] (-1.38,1.24) circle (0.5pt);
\fill [color=white] (-0.46,0.58) circle (0.5pt);
\fill [color=white] (-1.04,0.74) circle (0.5pt);
\fill [color=white] (-1.5,0.2) circle (0.5pt);
\fill [color=white] (-1.48,3.12) circle (0.5pt);
\fill [color=white] (-1.86,1.84) circle (0.5pt);
\fill [color=white] (-1.94,2.48) circle (0.5pt);
\fill [color=white] (-2.58,2.66) circle (0.5pt);
\fill [color=white] (-1.32,4.94) circle (0.5pt);
\fill [color=white] (-2.1,3.9) circle (0.5pt);
\fill [color=white] (-1.84,4.54) circle (0.5pt);
\fill [color=white] (-2.34,4.94) circle (0.5pt);
\fill [color=white] (-0.56,5.38) circle (0.5pt);
\fill [color=white] (0.72,5.12) circle (0.5pt);
\fill [color=white] (0.14,5.38) circle (0.5pt);
\fill [color=white] (0.22,6.14) circle (0.5pt);
\fill [color=white] (3.42,5.62) circle (0.5pt);
\fill [color=white] (4.16,5.1) circle (0.5pt);
\fill [color=white] (3.86,5.54) circle (0.5pt);
\fill [color=white] (2.2,5.52) circle (0.5pt);
\fill [color=white] (1.66,5.48) circle (0.5pt);
\fill [color=white] (2.78,5.5) circle (0.5pt);
\fill [color=white] (2.56,-1.14) circle (0.5pt);
\fill [color=white] (4.22,0.28) circle (0.5pt);
\fill [color=white] (2.2,6.14) circle (0.5pt);
\fill [color=white] (4.1,6.02) circle (0.5pt);
\end{small}
\end{tikzpicture}
\caption{The complementary k-NN set $A_k(\vec{i})$ illustrated for $\delta=1$ and $k=5$. Here we have $\vec{i}=(i_0,i_1)$. The vertices $i_{0}, i_{1}$ 
and their $k$-NNs 
are depicted in black and blue respectively. 
The complement of the union of $\{i_0,i_1\}$ and its $k$-NNs is the complementary $k$-NN set 
$A_k(\vec{i})$ and is 
depicted in red.}
\label{fig:KNNGraph1}
\end{center}
\end{figure}

Next we bound the term $b_3$ in (\ref{eq:chenstein}). The set
\be
A_k(\vec{i})=  \Beta_{\vec{i}}^c - \{\vec{i}\}
\label{eq:CompKNN}
\ee
indexes the complementary $k$-NN of $\bU_{\vec{i}}$ (see Fig. \ref{fig:KNNGraph1}) so that, using the
representation   (\ref{eq:thetamoms2}), \ben b_3 &=&
\sum_{\vec{i}\in {\mathcal C}^{<}} \mathbb E\left[\mathbb E\left[
\left.\prod_{l=1}^{\delta}\phi_{i_0i_l}-p_{\vec{i}}\right|\bU_{A_k(\vec{i})}\right]\right]
\\
&=&  \sum_{\vec{i}\in {\mathcal C}^{<}}
\int_{S_{2m-2}^{|A_k(\vec{i})|}} d\bfu_{A_k(\vec{i})}
\left(\prod_{l=1}^{\delta} \int_{S_{2m-2}}
d\bfu_{i_0} \int_{A(r,\bfu_{i_0})} d\bfu_{i_l}\right) \\ && 
\left( \frac{f_{\bfU_{\vec{i}}
|\bfU_{A_k}}(\bfu_{\vec{i}}|\bfu_{A_k(\vec{i})})-
f_{\bfU_{\vec{i}}}(\bfu_{\vec{i}})}{f_{\bfU_{\vec{i}}}(\bfu_{\vec{i}})}
\right) 
f_{\bfU_{\vec{i}}}(\bfu_{\vec{i}})f_{\bfU_{A_k(\vec{i})}}(\bfu_{A_k(\vec{i})})
\\
&\leq& O\left(p^{\delta+1}  P_0^{\delta}
\|\Delta_{p,m,k,\delta}\|_1\right) =
O\left(\eta_{p,\delta}^{\delta}\|\Delta_{p,m,k,\delta}\|_1\right).\een
Note that by definition of $\tilde{N}_{\delta,\rho}$ we have $\tilde{N}_{\delta,\rho}>0$ if and only if $N_{\delta,\rho}>0$. This yields: \be && \left|\mathbb{P}(N_{\delta,\rho}
>0)-\left(1-\exp(-\Lambda)\right)\right| \leq
\left|\mathbb{P}(\tilde{N}_{\delta,\rho} >0)-\mathbb{P}(N_{\delta,\rho} >0)\right|+ \nonumber \\ &&
\left|\mathbb{P}(\tilde{N}_{\delta,\rho} >0)-\left(1-\exp(-\mathbb E[\tilde{N}_{\delta,\rho}])\right)\right|
 +
\left|\exp(-\mathbb E[\tilde{N}_{\delta,\rho}])-\exp(-\Lambda)\right| \nonumber \\
&\leq & b_1+b_2+b_3 +
O\left(\left|\mathbb E[\tilde{N}_{\delta,\rho}]-\Lambda \right|\right)
\label{eq:allterms} \ee Combining the above inequalities on $b_1$,
$b_2$ and $b_3$ yields the first three terms in the argument of
the ``max'' on the right side of (\ref{eq:Pvlim}).

It remains to bound the  term
$|\mathbb E[\tilde{N}_{\delta,\rho}]-\Lambda|$. Application of
the mean value theorem to the multiple integral
(\ref{eq:thetamoms2}) gives \ben
\left|\mathbb E\left[\prod_{l=1}^{\delta}\phi_{ii_l}\right] - P_0^{\delta}
J\left(f_{\bfU_{i_1},\ldots,
\bfU_{i_{\delta}},\bfU_{i}}\right)\right| &\leq&
O\left(P_0^{\delta} r \right). \label{eq:phi_ijdbnd} \een
Applying relation (\ref{eq:rep1}) yields
%
\ben \left|\mathbb E[\tilde{N}_{\delta,\rho}] - p{p-1 \choose \delta}
P_0^{\delta} J\left(\ol{f_{\bU_{\ast1}, \ldots,
\bU_{\ast(\delta+1)}}}\right)\right| &\leq& O\left(p^{\delta+1}
P_0^{\delta} r\right) = O\left(\eta_{p,\delta}^{\delta} r\right).
\label{eq:phi_ijdbnd2} \een Combine this with (\ref{eq:allterms})
to obtain the bound (\ref{eq:Pvlim}). This completes the proof of
Theorem \ref{prop:parcor}. \qed
\end{proof}

An immediate consequence of Theorem \ref{prop:parcor} is the following result, similar to Proposition 2 in \cite{hero2012hub}, which provides asymptotic expressions for the mean number of $\delta$-hubs and the probability of the event $N_{\delta,\rho}>0$ as $p$ goes to $\infty$ and $\rho$ converges to $1$ at a prescribed rate. 

\begin{coro}
Let $\rho_p\in[0,1]$ be a sequence converging to one as $p\rightarrow
\infty$ such that $\eta_{p,\delta} = p^{1/\delta}(p-1)(1-\rho_p^2)^{(m-2)} \rightarrow
e_{m,\delta}\in (0,\infty)$. Then
\be
\lim_{p\rightarrow\infty}\mathbb E[N_{\delta,\rho_p}] = \Lambda_\infty =e_{m,\delta}^\delta/\delta! \;  \lim_{p\rightarrow \infty}
J(\ol{f_{\bU_{\ast1}, \ldots, \bU_{\ast(\delta+1)}}}).
\label{eq:ENdef}
\ee
Assume that $k=o(p^{1/\delta})$ and that for the weak dependency coefficient $\|\Delta_{p,m,k,\delta}\|_1$, defined via \eqref{eq:Deltapdefavg}, we have
$\lim_{p\rightarrow\infty} \|\Delta_{p,m,k,\delta}\|_1=0$. Then
\be
\mathbb{P}(N_{\delta,\rho_p}>0)\rightarrow
1-\exp(-\Lambda_\infty / \varphi(\delta)).
\label{eq:Poissonconv}
\ee

\label{prop:parcor1}
\end{coro}


Corollary \ref{prop:parcor1} shows that in the limit $p \to \infty$, the number of detected hubs depends on the true population correlations only through the quantity $J(\ol{f_{\bU_{\ast1}, \ldots, \bU_{\ast(\delta+1)}}})$.  In some cases $J(\ol{f_{\bU_{\ast1}, \ldots, \bU_{\ast(\delta+1)}}})$ can be evaluated explicitly. Similar to the argument in \cite{hero2012hub}, it can be shown that if the population covariance matrix $\mathbf \Sigma$ is sparse in the sense that its non-zero off-diagonal entries can be arranged into a $k \times k$ submatrix by reordering rows and columns, then 
\ben
J(\ol{f_{\bU_{\ast1}, \ldots, \bU_{\ast(\delta+1)}}}) = 1 + O(k/p).
\een
Hence, if $k=o(p)$ as $p \rightarrow \infty$, the quantity $J(\ol{f_{\bU_{\ast1}, \ldots, \bU_{\ast(\delta+1)}}})$ 
converges to $1$. If $\mathbf \Sigma$ is diagonal, then $J(\ol{f_{\bU_{\ast1}, \ldots, \bU_{\ast(\delta+1)}}}) = 1$ exactly. In such cases, the quantity $\Lambda_{\infty}$ in Corollary \ref{prop:parcor1} does not depend on the unknown underlying distribution of the U-scores. As a result, the expected number of $\delta$-hubs in $\mathcal G_{\rho}({\mathbf \Psi})$ and the probability of discovery of at least one $\delta$-hub 
do not depend on the underlying distribution. We will see in Sec. \ref{sec:MultipleInf} that this result is useful in assigning statistical significance levels 
to vertices of the graph $\mathcal G_{\rho}({\mathbf \Psi})$.


\subsection{Phase Transitions and Critical Threshold}
\label{subsec:corrScreenPhase}

It can be seen from Theorem \ref{prop:parcor} and Corollary \ref{prop:parcor1} that 
the number of $\delta$-hub discoveries exhibits a phase transition in the high-dimensional regime where the number of variables $p$ can be very large relative to the number of samples $m$. Specifically, assume that the population covariance matrix $\mathbf\Sigma$ is block-sparse as in Section \ref{sec:Theoretical}.  Then as the correlation threshold $\rho$ is reduced, 
the number of $\delta$-hub discoveries abruptly increases to the maximum, $p$. Conversely as 
$\rho$ increases, the number of discoveries 
quickly approaches zero. Similarly, the family-wise error rate 
(i.e.\ the probability of discovering at least one $\delta$-hub in a graph with no true hubs) exhibits a phase transition as a function of 
$\rho$. Figure \ref{fig:FalseDiscoveryp1000d1} shows the family-wise error 
rate obtained via expression \eqref{eq:Poissonconv} for $\delta=1$ and $p=1000$, as a function of 
$\rho$ and the number of samples $m$. It is 
seen that for a fixed value of $m$ there is a sharp transition in the family-wise error 
rate as a function of $\rho$.
\begin{figure}
\begin{center}
\includegraphics[width=9cm]{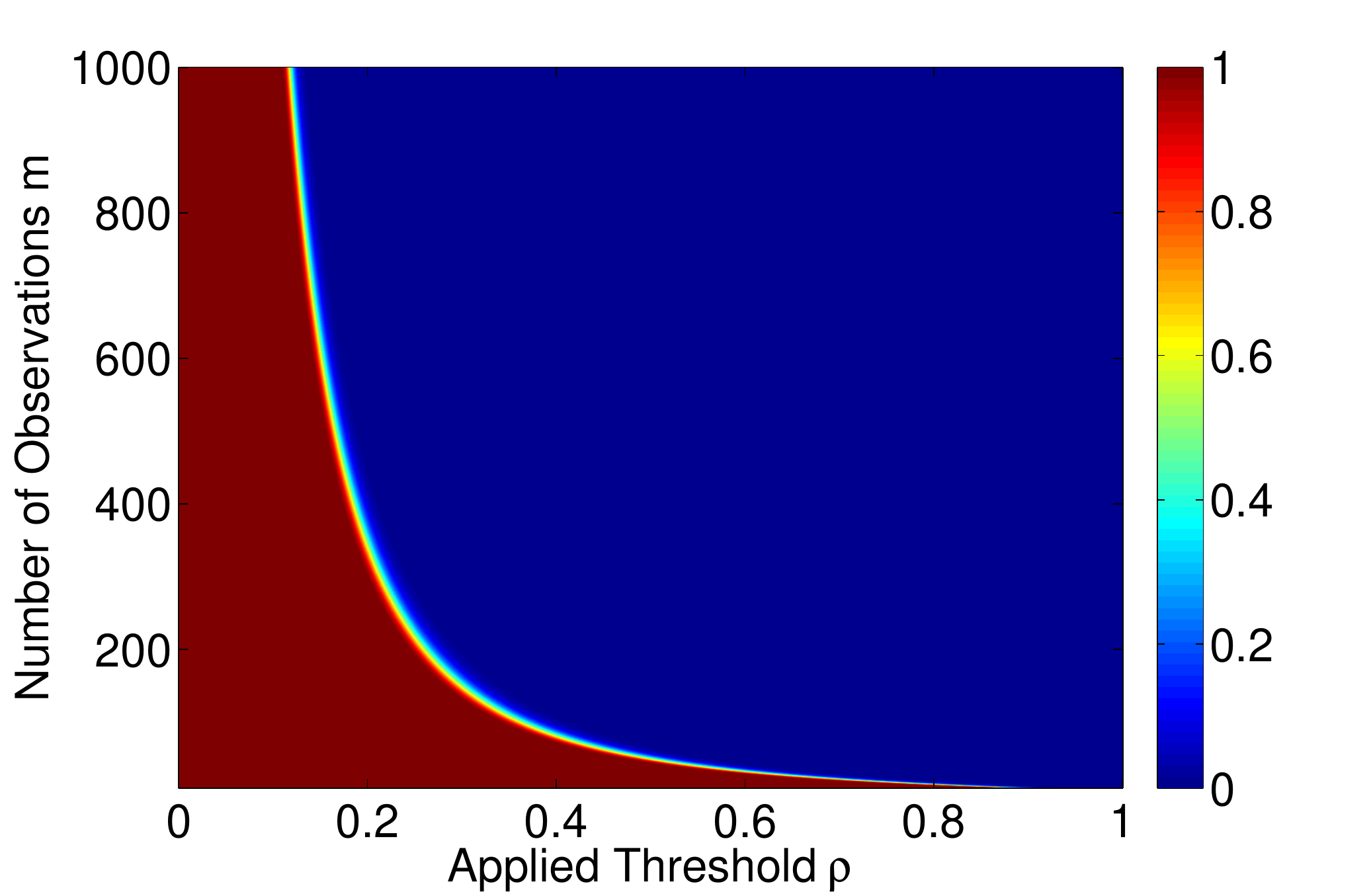}
\end{center}
\caption{
Family-wise error rate as a function of correlation threshold $\rho$ and number of samples $m$ for $p=1000, \delta=1$. The phase transition phenomenon is clearly observable in the plot.}
\label{fig:FalseDiscoveryp1000d1}
\end{figure}

The phase transition phenomenon motivates the definition of a critical threshold $\rho_{c,\delta}$ as the threshold $\rho$ satisfying the following slope condition: 
\ben
\partial \mathbb E[N_{\delta,\rho}]/ \partial \rho = -p.
\label{eq:CricThrshldDef}
\een
Using \eqref{eq:Lambdadef} the solution of the above equation can be approximated via the expression below:
\be
\rho_{c,\delta} = \sqrt{1-(c_{m,\delta}(p-1))^{-2\delta/(\delta(2m-3)-2)}},
\label{CriticalThreshold}
\ee
where $c_{m,\delta} = b_{m-1} \delta J(\ol{f_{\bU_{\ast1}, \ldots, \bU_{\ast(\delta+1)}}})$. 
The screening threshold $\rho$ should be chosen 
greater than $\rho_{c,\delta}$ to prevent excessively 
large numbers of false positives. 
Note that the critical threshold $\rho_{c,\delta}$ also does not depend on the underlying distribution of the U-scores when the covariance matrix $\mathbf\Sigma$ is block-sparse.  

Expression \eqref{CriticalThreshold} is similar to the expression obtained in \cite{hero2012hub} for the critical threshold in real-valued correlation screening. However, in the complex-valued case 
the coefficient $c_{m,\delta}$ and the exponent of the term $c_{m,\delta}(p-1)$ are different from the real case. This generally results in smaller values of $\rho_{c,\delta}$ for fixed 
$m$ and $\delta$. 

Figure \ref{fig:Critical1} shows the value of $\rho_{c,\delta}$ obtained via \eqref{CriticalThreshold} as a function of $m$ for different values of $\delta$ and $p$.  The critical threshold
decreases as either the sample size $m$ increases, the number of variables $p$ decreases, or the
vertex degree $\delta$ increases. Note that even for ten billion ($10^{10}$) dimensions (upper triplet of curves in the
figure) only a relatively small number of samples are necessary for complex-valued correlation screening to be useful. For example, with $m = 200$ one can reliably discover connected vertices ($\delta= 1$ in the
figure) having correlation greater than $\rho_{c,\delta}= 0.5$.

\begin{figure}
\begin{center}
\includegraphics[width=9cm]{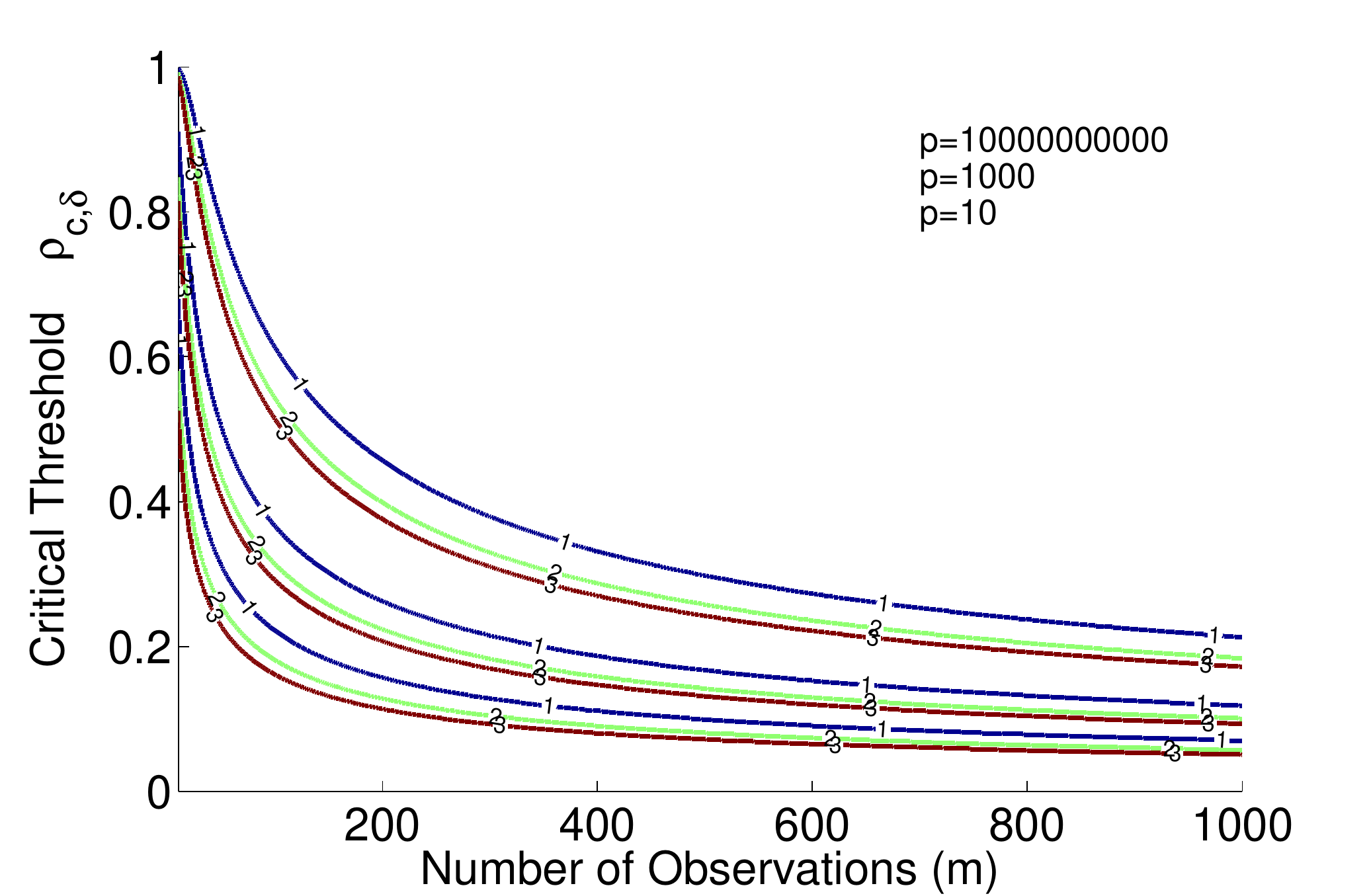}
\end{center}
\caption{The critical threshold $\rho_{c,\delta}$ as a function of the sample size $m$ for $\delta=1, 2, 3$ (curve labels) and $p=10, 1000, 10^{10}$ (bottom to top triplets of curves). The figure shows that the critical threshold decreases as either $m$ or $\delta$ increases. When the number of samples $m$ is small the critical threshold is close to $1$ in which case reliable hub discovery is impossible. However  a relatively small increment in $m$ is sufficient to reduce the critical threshold significantly. For example for $p=10^{10}$, only $m=200$ samples are enough to bring $\rho_{c,1}$ down to $0.5$.}
\label{fig:Critical1}
\end{figure}

\section{Application to Spectral Screening of Multivariate Gaussian Time Series}
\label{sec:MultipleInf}

In this section, the complex-valued correlation hub screening method of Section \ref{sec:corrScreen} is applied to stationary multivariate Gaussian time series. 
Assume that the time series $X^{(1)},\cdots,X^{(p)}$ defined in Section \ref{GaussianWSS} satisfy the conditions of Corollary \ref{cor:independence}. 
Assume also that a total of $N = n \times m$ time samples of $X^{(1)},\cdots,X^{(p)}$ are available. We divide the $N$ samples into $m$ parts of $n$ consecutive samples and we take the $n$-point DFT of each part. Therefore, for each time series, at each frequency $f_i = (i-1) /n$, $1 \leq i \leq n$, $m$ samples are available. This allows us to construct a (partial) correlation graph corresponding to each frequency. We denote the (partial) correlation graph corresponding to frequency $f_i$ and correlation threshold $\rho_i$ as $\mathcal G_{f_i,\rho_i}$. $\mathcal G_{f_i,\rho_i}$ has $p$ vertices $v_1,v_2,\cdots,v_p$ corresponding to time series $X^{(1)},X^{(2)},\cdots,X^{(p)}$, respectively. Vertices $v_k$ and $v_l$ are connected if the magnitude of the sample (partial) correlation between the DFTs of $X^{(k)}$ and $X^{(l)}$ at frequency $f_i$ (i.e. the sample (partial) correlation between $Y^{(k)}(i-1)$ and $Y^{(l)}(i-1)$) is at least $\rho_i$.

Consider a single frequency $f_{i}$ and the null hypothesis, $\mathcal{H}_0$, that the 
correlations among the time series $X^{(1)},X^{(2)},\cdots,X^{(p)}$ at frequency $f_{i}$ are block sparse 
in the sense of Section \ref{sec:Theoretical}.  
As discussed in Sec.~\ref{sec:Theoretical}, under $\mathcal{H}_{0}$ 
the expected number of $\delta$-hubs and the probability of discovery of at least one $\delta$-hub in graph $\mathcal G_{f_{i},\rho_i}$ are not functions of the unknown underlying distribution of the data. 
Therefore the results of Corollary \ref{prop:parcor1} may be used to quantify the statistical significance of declaring 
vertices of $\mathcal G_{f_{i},\rho_i}$ to be $\delta$-hubs.  The statistical significance is represented by the 
p-value, defined in general as the probability of having 
a test statistic at least as extreme as the value actually observed 
assuming that the null hypothesis $\mathcal{H}_{0}$ is true. 
In the case of correlation hub screening, the p-value $pv_{\delta}(j)$ assigned to vertex $v_j$ for being a $\delta$-hub is the maximal probability that $v_j$ maintains degree $\delta$ given the observed sample correlations, 
assuming that the block-sparse hypothesis $\mathcal{H}_{0}$ is true.
The detailed procedure for assigning p-values 
is similar to the procedure in \cite{hero2012hub} for real-valued correlation screening and is illustrated in Fig.~\ref{fig:pvals}. Equation \eqref{CriticalThreshold} helps in choosing  the initial threshold $\rho^*$.

\begin{figure}[!h]

\vrule
\begin{minipage}{11.5cm} 
   \centering
\hrule \vspace{0.5em} 
   \begin{minipage}{11cm}    
        \begin{itemize}
         {\setlength\itemindent{5pt}\item Initialization:

           \begin{enumerate}
			
			 \item Choose a degree threshold $\delta \geq 1$.
				
             \item Choose an initial threshold $\rho^*>\rho_{c,\delta}$.

             \item Calculate the degree $d_j$ of each vertex of graph ${\mathcal G}_{\rho^*}({\mathbf \Psi})$.
			\item Select a value of $\delta \in \{1,\cdots,\max_{1 \leq j \leq p} d_j\}$.
			
           \end{enumerate}

        	\item For each $j = 1,\cdots, p$ find $\rho_j(\delta)$ as the $\delta$th greatest element of the $j$th row of the sample (partial) correlation matrix.
        	
        	\item Approximate the p-value corresponding to vertex $v_{j}$ 
	as $pv_{\delta}(j) \approx 1-\exp(-\mathbb E[N_{\delta,\rho_j(\delta)}] / \varphi(\delta))$, where $\mathbb E[N_{\delta,\rho_j(\delta)}]$ is approximated by 
	the limiting expression \eqref{eq:ENdef} using $J(\ol{f_{\bU_{\ast1}, \ldots, \bU_{\ast(\delta+1)}}})=1$.

\item Screen variables by thresholding the p-values $pv_\delta(j)$ at desired significance level.
}
        \end{itemize}

\end{minipage}
\vspace{0.5em} \hrule
\end{minipage}\vrule \\
\\
\caption{Procedure for assigning p-values to the vertices of ${\mathcal G}_{\rho^*}({\mathbf 			\Psi})$.}
\label{fig:pvals}
\end{figure}

Given Corollary \ref{cor:independence}, for $i \neq j$ the correlation graphs $\mathcal G_{f_i,\rho_i}$ and $\mathcal G_{f_j,\rho_j}$ and their associated inferences are approximately independent.  Thus we can solve multiple inference problems by first performing correlation hub screening on each graph as discussed above and then aggregating the inferences at each frequency in a straightforward manner.  Examples of aggregation procedures are described below.

\subsection{Disjunctive Hubs}
One task that can be easily performed 
is finding the p-value for a given time series 
to be a hub in at least one of the graphs $\mathcal G_{f_1,\rho_1}, \cdots, \mathcal G_{f_n,\rho_n}$. More specifically, for each $j = 1, \ldots, p$ 
 denote the p-values for vertex $v_j$ being a $\delta$-hub in $\mathcal G_{f_1,\rho_1}, \cdots, \mathcal G_{f_n,\rho_n}$ by $pv_{f_1,\rho_1,\delta}(j),\cdots,pv_{f_n,\rho_n,\delta}(j)$ respectively.  These p-values are obtained using the method of Fig. \ref{fig:pvals}.  Then $pv_{\delta}(j)$, the p-value for the vertex $v_j$ being a $\delta$-hub in at least one of the frequency graphs $\mathcal G_{f_1,\rho_1}, \cdots, \mathcal G_{f_n,\rho_n}$ can be approximated as:
\ben
\mathbb{P}(\exists i: d_{j,f_i} \geq \delta ~|\mathcal{H}_0) \approx \hat{pv}_{\delta}(j) = 1-\prod_{i=1}^{n}(1-pv_{f_i,\rho_i,\delta}(j)),
\een
in which $d_{j,f_i}$ is the degree of $v_j$ in the graph $\mathcal G_{f_i,\rho_i}$.
\subsection{Conjunctive Hubs}
Another property of interest is the existence of a hub at all frequencies for a particular time series. 
In this case we have:
\ben
\mathbb{P}(\forall i: d_{j,f_i} \geq \delta ~| \mathcal{H}_0) \approx \check{pv}_{\delta}(j) = \prod_{i=1}^{n}pv_{f_i,\rho_i,\delta}(j).
\een
\subsection{General Persistent Hubs}
The general case is the event that at least $K$ frequencies have hubs of degree at least $\delta$ at vertex $v_j$. For this general case we have:
\ben
&&\mathbb{P}(\exists i_1,\ldots,i_{K} : d_{j,f_{i_1}} \geq \delta, \ldots d_{j,f_{i_{K}}} \geq \delta ~| \mathcal{H}_0) = \\&& \sum_{k'=K}^{n} \sum_{\stackrel{i_1< \ldots< i_{k'}, i_{k'+1} < \ldots < i_n }{\{i_1,\ldots,i_n\} = \{1, \ldots, n\}}} \prod_{l=1}^{k'} pv_{f_{i_l},\rho_{i_l},\delta}(j) \prod_{l'=k'+1}^{n} \left(1 - pv_{f_{i_{l'}},\rho_{i_{l'}},\delta}(j)\right).
\een

\section{Experimental Results}
\label{sec:Sims}
\subsection{Phase Transition Phenomenon and Mean Number of Hubs}
We first performed numerical simulations to confirm Theorem \ref{prop:parcor} and Corollary \ref{prop:parcor1} for complex-valued correlation screening.
Samples were generated from $p$ uncorrelated complex Gaussian random variables. 
Figure \ref{fig:PhaseTrans} shows the number of discovered $1$-hubs 
for $p=1000$ and several sample sizes $m$. The plots from left to right correspond to $m = 2000, 1000, 500, 100, 50, 20, 10, 6$ and $4$, respectively.  The phase transition phenomenon is clearly observed in the plot. Table \ref{table:PTpredict} shows the predicted value obtained from formula \eqref{CriticalThreshold} for the critical threshold. As can be seen in Fig. \ref{fig:PhaseTrans}, the empirical phase transition thresholds approximately match the predicted values of Table \ref{table:PTpredict}. 
Moreover, to confirm the accuracy of equation \eqref{eq:ENdef} in Corollary \ref{prop:parcor1}, we list the number of hubs for $m=100$ in Table \ref{table:EN}. The left column shows the empirical average number of hubs of degree at least $\delta =1,2,3,4$ in a network of i.i.d. complex Gaussian random variables. The numbers in this column are obtained by averaging $1000$ independent experiments. The right column shows the predicted value of $\mathbb E[N_{\delta,\rho}]$ obtained via formula \eqref{eq:ENdef} with $J(\ol{f_{\bU_{\ast1}, \ldots, \bU_{\ast(\delta+1)}}})=1$ for the i.i.d. case. 
As we see the empirical and predicted values are close to each other.

\begin{figure}[htb]
\begin{center}
\includegraphics[height=7cm]{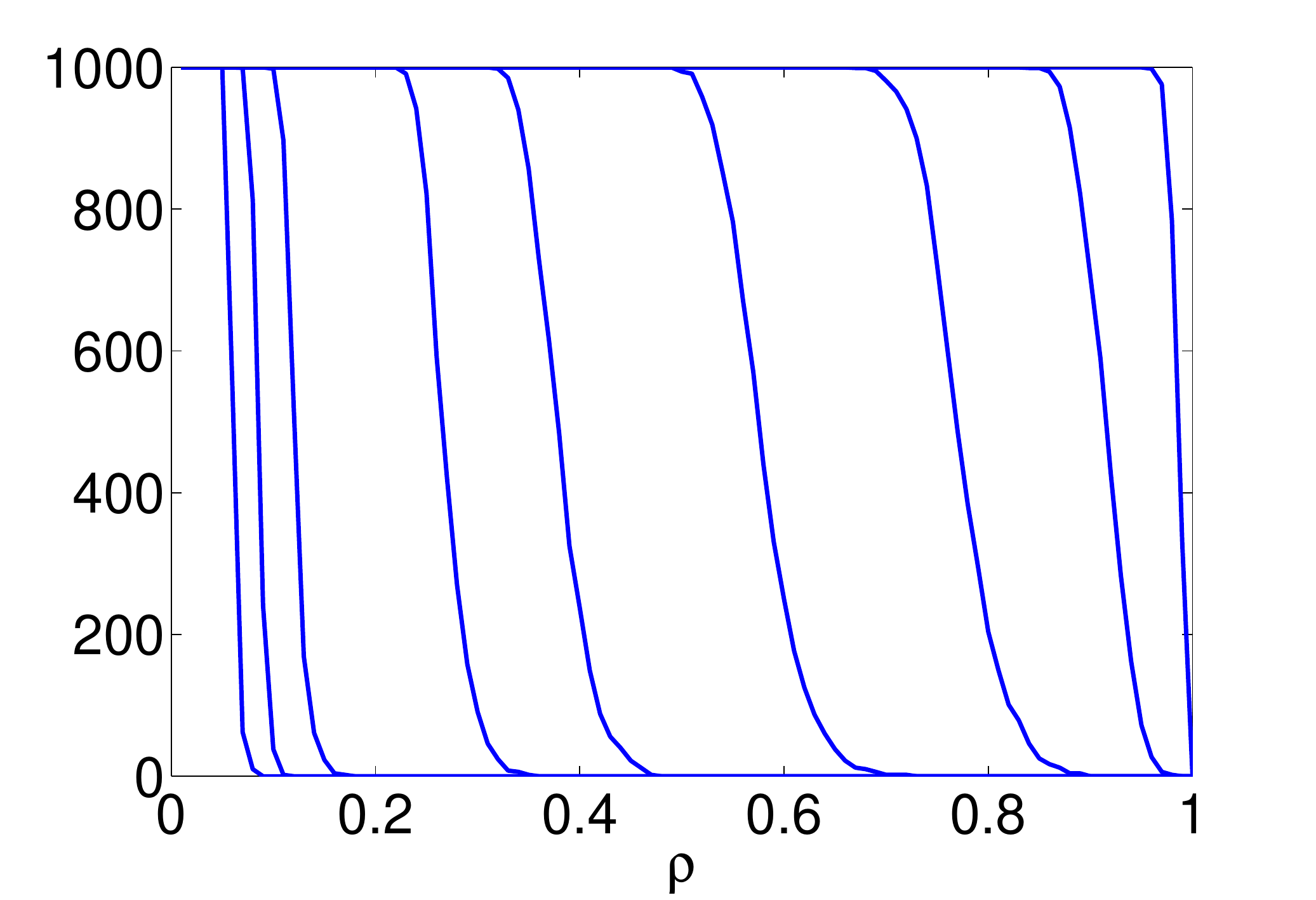}
\end{center}
\caption{Phase transition phenomenon: the number of $1$-hubs in the sample correlation graph corresponding to uncorrelated complex Gaussian variables as a function of correlation threshold $\rho$. Here, $p=1000$ and the plots from left to right correspond to $m = 2000, 1000, 500, 100, 50, 20, 10, 6$ and $4$, respectively.}
\label{fig:PhaseTrans}
\end{figure}

\begin{table}
\begin{center}
\begin{tabular}{|c|c|c|c|c|c|c|c|c|c|}  \hline
$m$ & 2000 & 1000 & 500 & 100 & 50 & 20 & 10 & 6 & 4  \\ \hline
$\rho_{c,\delta}$ & 0.05  &  0.07  &  0.10  &  0.24  &  0.35  &  0.56  &  0.78  &  0.94  &  0.99  \\ \hline
\end{tabular}
\caption{The value of critical threshold $\rho_{c,\delta}$ obtained from formula \eqref{CriticalThreshold} for $p=1000$ complex variables and $\delta=1$. The predicted $\rho_{c,\delta}$ approximates the 
phase transition thresholds in Fig. \ref{fig:PhaseTrans}.}
\label{table:PTpredict}
\end{center}
\end{table}

\begin{table}[!h]
\begin{center}
\begin{tabular}{|c||c|c|}  \hline
degree threshold &  empirical ($\mathbb E[N_{\delta,\rho}]$)  & predicted ($\mathbb E[N_{\delta,\rho}]$) \\ \hline
$d_i\geq \delta =1$ & 284  & 335 \\ \hline
$d_i\geq \delta =2$ & 45 & 56 \\ \hline
$d_i\geq \delta =3$ & 5  & 6 \\ \hline
$d_i\geq \delta =4$ & 0  & 0 \\ \hline
\end{tabular}
\caption{Empirical average number of discovered hubs vs. predicted average number of discovered hubs in an uncorrelated complex Gaussian network. Here $p=1000$, $m=100$, $\rho = 0.28$. The empirical values are obtained by performing $1000$ independent experiments.}
\label{table:EN}
\end{center}
\end{table}

\subsection{Asymptotic Independence of Spectral Components 
for AR(1) Model}
To illustrate the asymptotic independence property and convergence rate of Theorem \ref{thm:independence}, we considered the simple case of an AR(1) process,
\be
X(k) = \varphi_1 X(k-1) + \varepsilon(k), ~~~~~ k \geq 1,
\label{eq:AR1Example}
\ee
in which $X(0) = 0, \varphi_1 = 0.9$ and $\varepsilon(.)$ is a stationary Gaussian 
process with no temporal correlation and standard deviation $1$. We performed Monte-Carlo simulations to compute the correlation between 
spectral components at different frequencies for window sizes $n = 10, 20, \ldots, 250$. More specifically, we set $k = 1$ and $l = 2$ and empirically estimated 
$|\mathrm{cor}\left( Y(k),Y(l) \right)|$ using $50000$ Monte-Carlo trials for each value of window size $n$. Figure \ref{fig:AR1cor} shows the result of this experiment. It is observable that the magnitude of $\mathrm{cor}\left( Y(k),Y(l) \right)$ is bounded above by the function $10/n$. This observation is consistent with Theorem \ref{thm:independence}.

\begin{figure}[htb]
\begin{center}
\includegraphics[height=7cm]{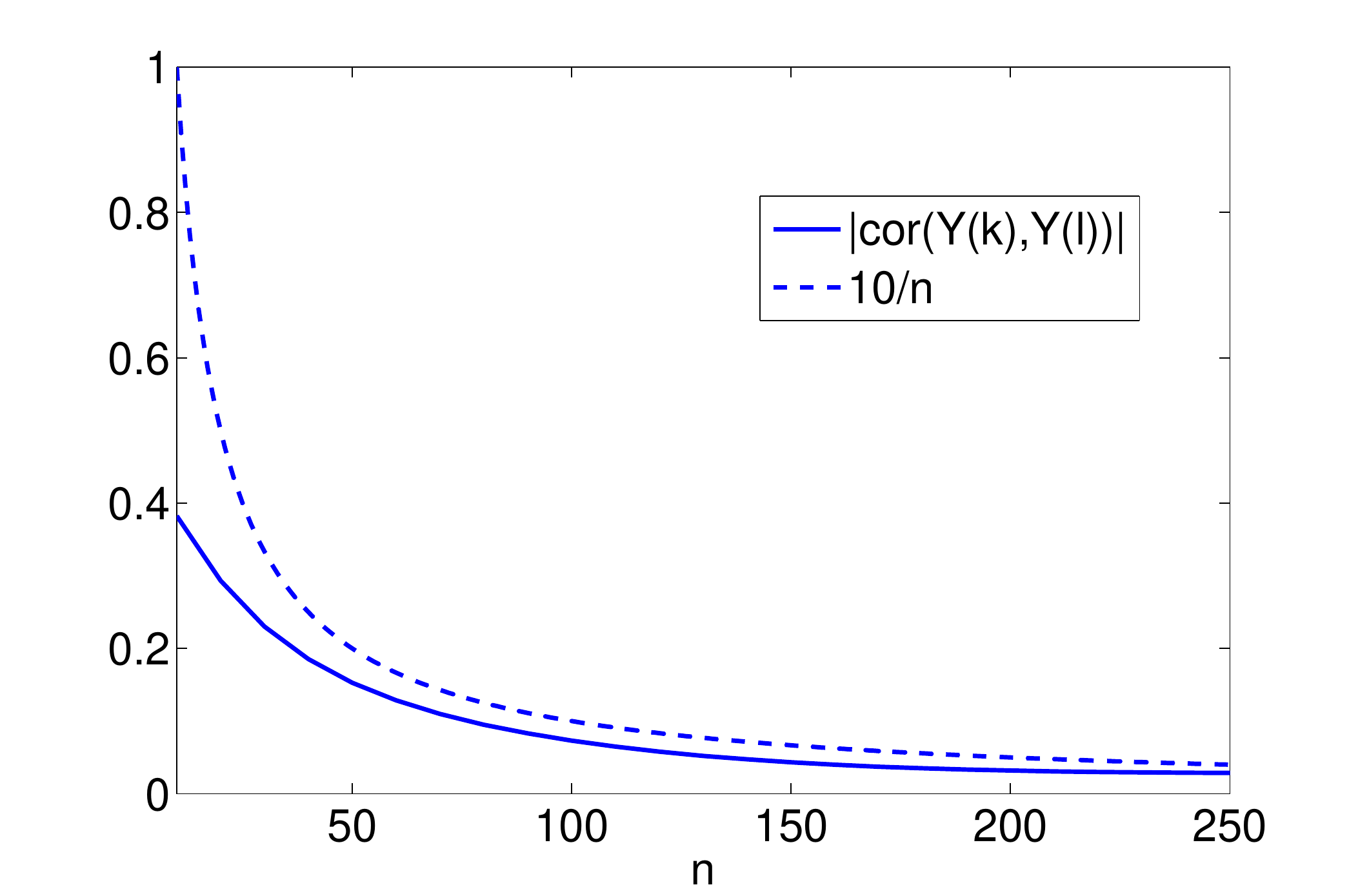}
\end{center}
\caption{Correlation coefficient $|\mathrm{cor}\left( Y(1),Y(2) \right)|$ as a function of window size $n$, 
empirically estimated using $50000$ Monte-Carlo trials. Here $Y(.)$ is the DFT of the AR($1$) process \eqref{eq:AR1Example}. 
The magnitude of the correlation for $n=10, 20, \ldots, 250$ is bounded above by the function $10/n$. This observation is consistent with the convergence rate 
in Theorem \ref{thm:independence}.}
\label{fig:AR1cor}
\end{figure}

\subsection{Spectral Correlation Screening of a Band-Pass Multivariate Time Series}

Next we analyzed the performance of the proposed complex-valued correlation screening framework on a synthetic data set for which the expected results are known.

We synthesized a multivariate stationary Gaussian time series using the the following procedure. Here we set $p=1000, N=12000$ and $m=n=100$. The discrepancy between $N$ and the product $mn$ is explained below. Let $X(k), 0 \leq k \leq N-1$ be a sequence of i.i.d.\ zero-mean Gaussian random variables (i.e.\ white Gaussian noise) with standard deviation of $1$. The $p$ time series $X^{(1)}(k),\ldots,X^{(p)}(k), 0 \leq k \leq N-1$ are obtained from $X(k)$ by band-pass filtering and adding independent white Gaussian noise.  Specifically,
\ben
X^{(i)}(k) = h_i(k) \conv X(k) + N_i(k),~~~~ 1 \leq i \leq p, 0 \leq k \leq N-1,
\een
in which $\conv$ represents the convolution operator, $h_i(.)$ is the impulse response of the $i$th band-pass filter and $N_i(.)$ is an independent white Gaussian noise series whose standard deviation is $0.1$.  Since stable filtering of a stationary series results in another stationary series, the obtained series $X^{(1)}(k),\ldots,X^{(p)}(k)$ are stationary and Gaussian. For $i = 10l, 1 \leq l \leq 50$, $h_i(k)$ is the impulse response of a band-pass filter with pass band $f \in [(4l-1)/400, 4l/400]$. We approximate the ideal band-pass filters with finite impulse response (FIR) Chebyshev filters \cite{oppenheim1989discrete}. 
Also for $i = 500 + 10l, 1 \leq l \leq 50$ we set $h_i(k) = h_{i-500}(k)$. For all of the other values of $i$ (i.e. $i \neq 10l$) we set $h_i(k)=0, 0 \leq k \leq N-1$. 

Figure \ref{fig:SampleSignals} shows the signal part of the time series (i.e. $h_i(k) \conv X(k)$) for $i = 100, 200, 300, 400$.  It is seen that the first $2000$ samples of the signals 
reflect the transient response of the filters. 
These $2000$ samples are not included for the purpose of correlation screening.  Hence the 
actual number of time samples considered is 
$mn=10000$. Figure \ref{fig:fftMag} shows the magnitude of the DFTs of the signals, $Y^{(i)}(k)$, for $i = 50, 100, \ldots, 500$. The band-pass structure of the signals is clearly observable in the figure.

\begin{figure}[htb]
\begin{center}
\includegraphics[height=7cm]{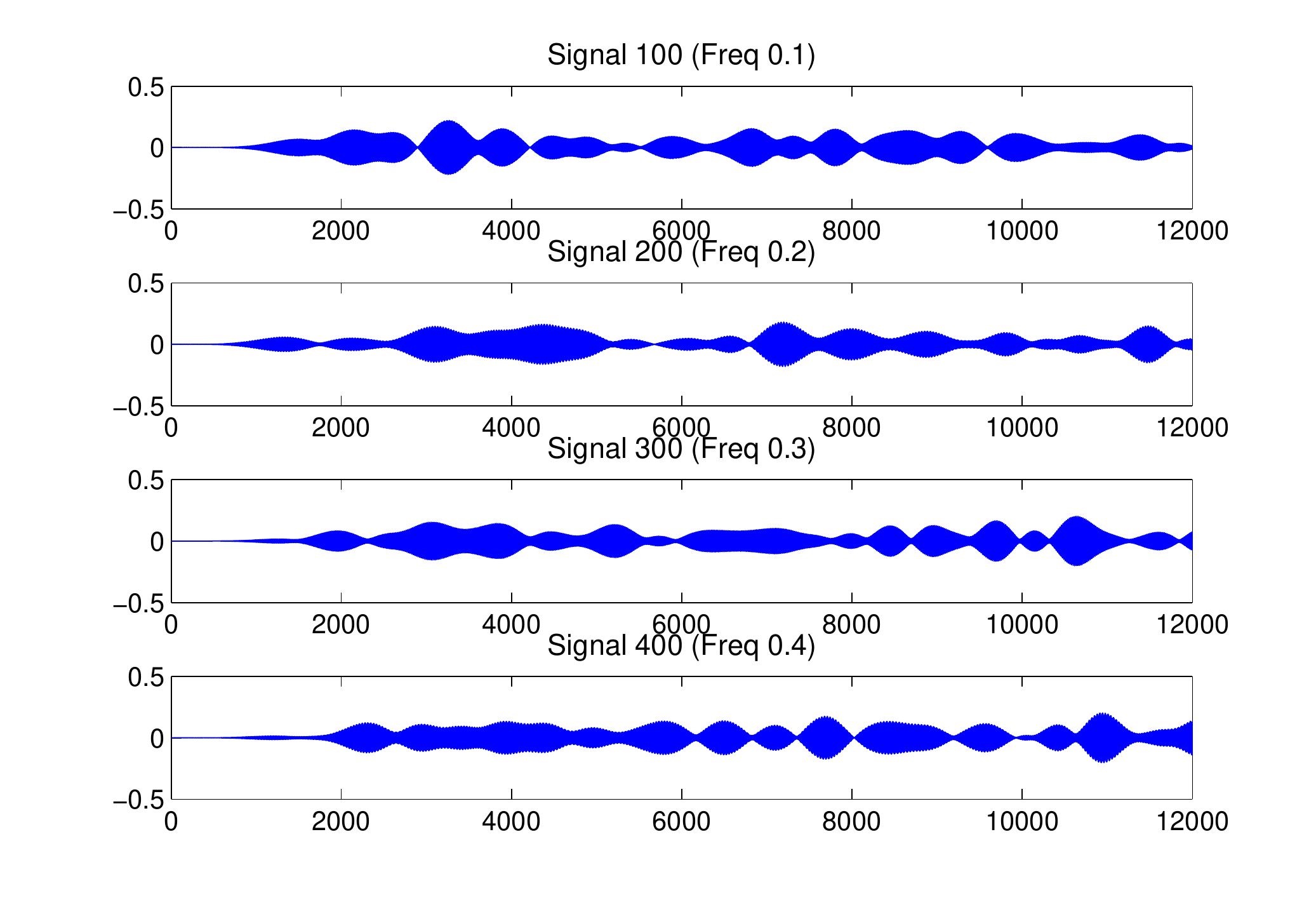}
\end{center}
\caption{Signal part of the band-pass time series $X^{(i)}(k)$ (i.e. $h_i(k) \conv X(k)$) for $i = 100, 200, 300, 400$.}
\label{fig:SampleSignals}
\end{figure}

\begin{figure}[!h]
\begin{center}
\includegraphics[height=7cm]{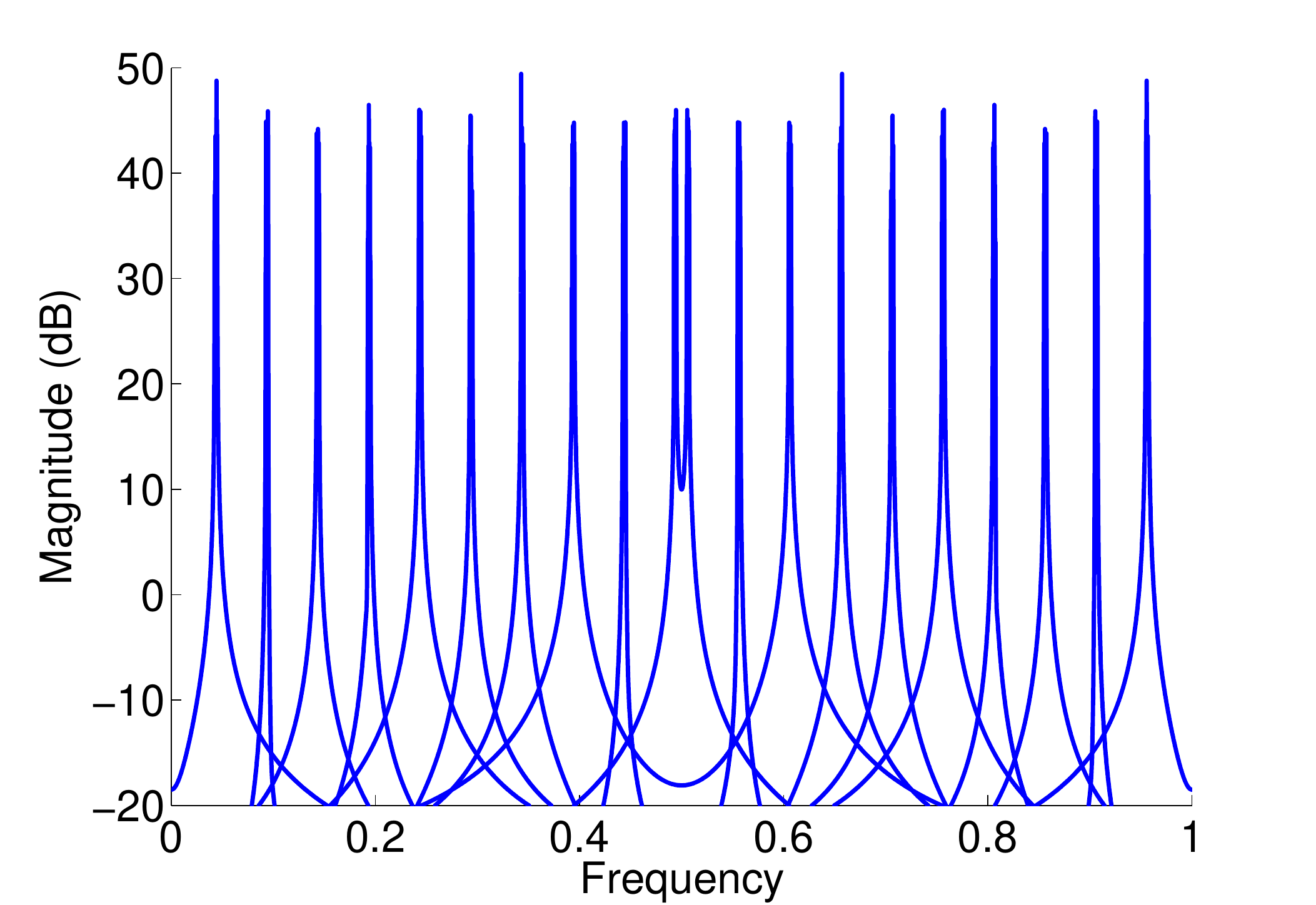}
\end{center}
\caption{
DFT magnitude of the band-pass signals $h_i(k) \conv X(k)$ 
(i.e. $20\log_{10}(|Y^{(i)}(.)|)$) as a function of frequency for $i=50, 100, \ldots, 500$.}
\label{fig:fftMag}
\end{figure}

\begin{figure}[htb]
\begin{center}
\includegraphics[height=4.55cm]{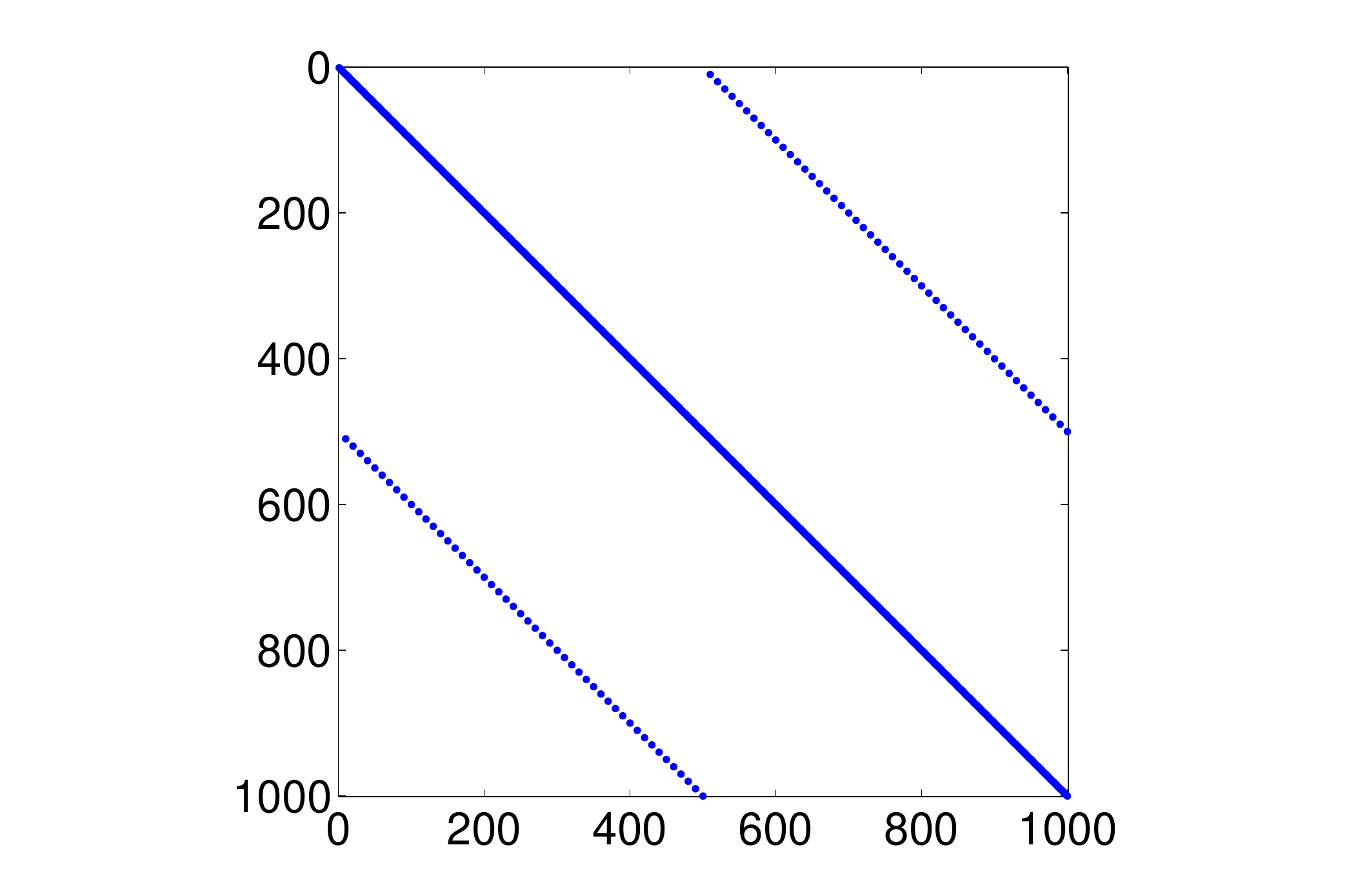}
\includegraphics[height=4.55cm]{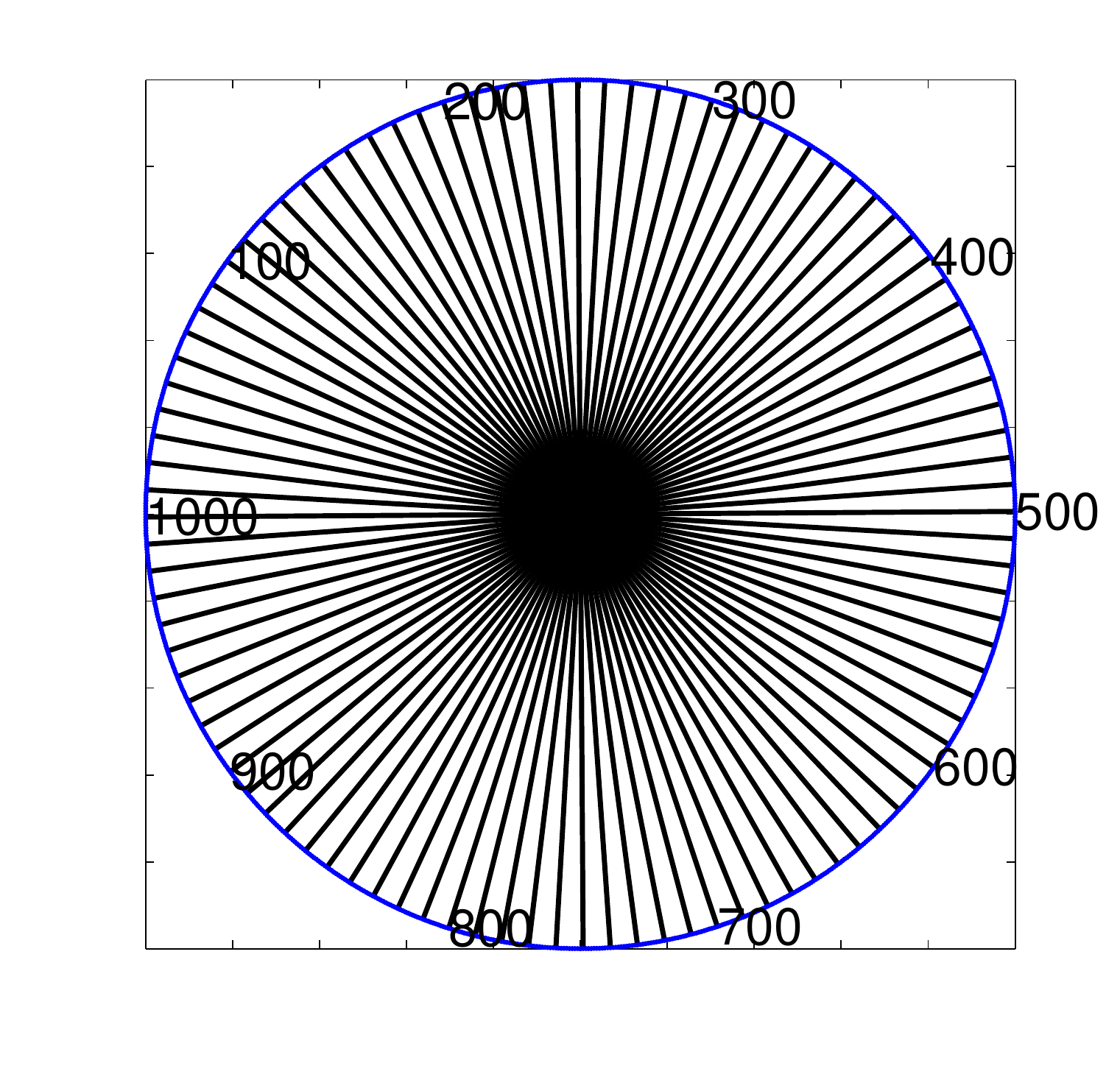}
\end{center}
\caption{(Left) The structure of the thresholded sample correlation matrix in the time domain. (Right) The correlation graph corresponding to the thresholded sample correlation matrix in the time domain.}
\label{fig:spyTime}
\end{figure}

We first constructed a correlation matrix for the time series $X^{(1)}(k),\ldots,$ $X^{(p)}(k)$ from their simultaneous time samples. Figure \ref{fig:spyTime} illustrates the structure of the thresholded sample correlation matrix and the corresponding correlation graph.  Note that this is a real-valued correlation screening problem in the time domain. The correlation threshold used here is $\rho = 0.2$ which is well above the critical threshold $\rho_{c,1} = 0.028$ obtained via formula ($10$) in \cite{hero2012hub} for $p=1000$ and $N = 10000$.

To examine the spectral structure of the correlations in Fig.~\ref{fig:spyTime}, we then performed complex-valued correlation screening on the spectra of the 
time series $X^{(1)}(k),\ldots,X^{(p)}(k)$. Figure \ref{fig:FreqGraphs2} shows the constructed correlation graphs $\mathcal G_{f,\rho}$ for $f=[0.1, 0.2, 0.3, 0.4]$ and correlation threshold $\rho = 0.9$, which corresponds to a $\delta=1$ false positive rate $\mathbb{P}(N_{\delta,\rho}>0) \approx 10^{-65}$ (using $\delta=1$ in the asymptotic relation \eqref{eq:Poissonconv} with $\Lambda_\infty=e_{m,\delta}^\delta/\delta !$ as specified by \eqref{eq:ENdef}). Note that the value of the correlation threshold is set to be higher than the critical threshold $\rho_c =0.24$. It can be observed that performing complex-valued spectral correlation screening at each frequency correctly discovers the correlations between the time series which are active around that frequency. As an example, for $f=0.2$ the discovered hubs (for $\delta = 1$) are the time series $X^{(i)}(k)$ for $i \in \{200, 700\}$. These time series are the ones that are active at frequency $f=0.2$. Under the null hypothesis of diagonal covariance matrices, the p-values for the discovered hubs are of order $10^{-65}$ or smaller. 
These results show that complex-valued spectral correlation screening is able to resolve the sources of correlation between time series in the spectral domain. 

\begin{figure}[!h]
\begin{center}
\includegraphics[height=5.15cm]{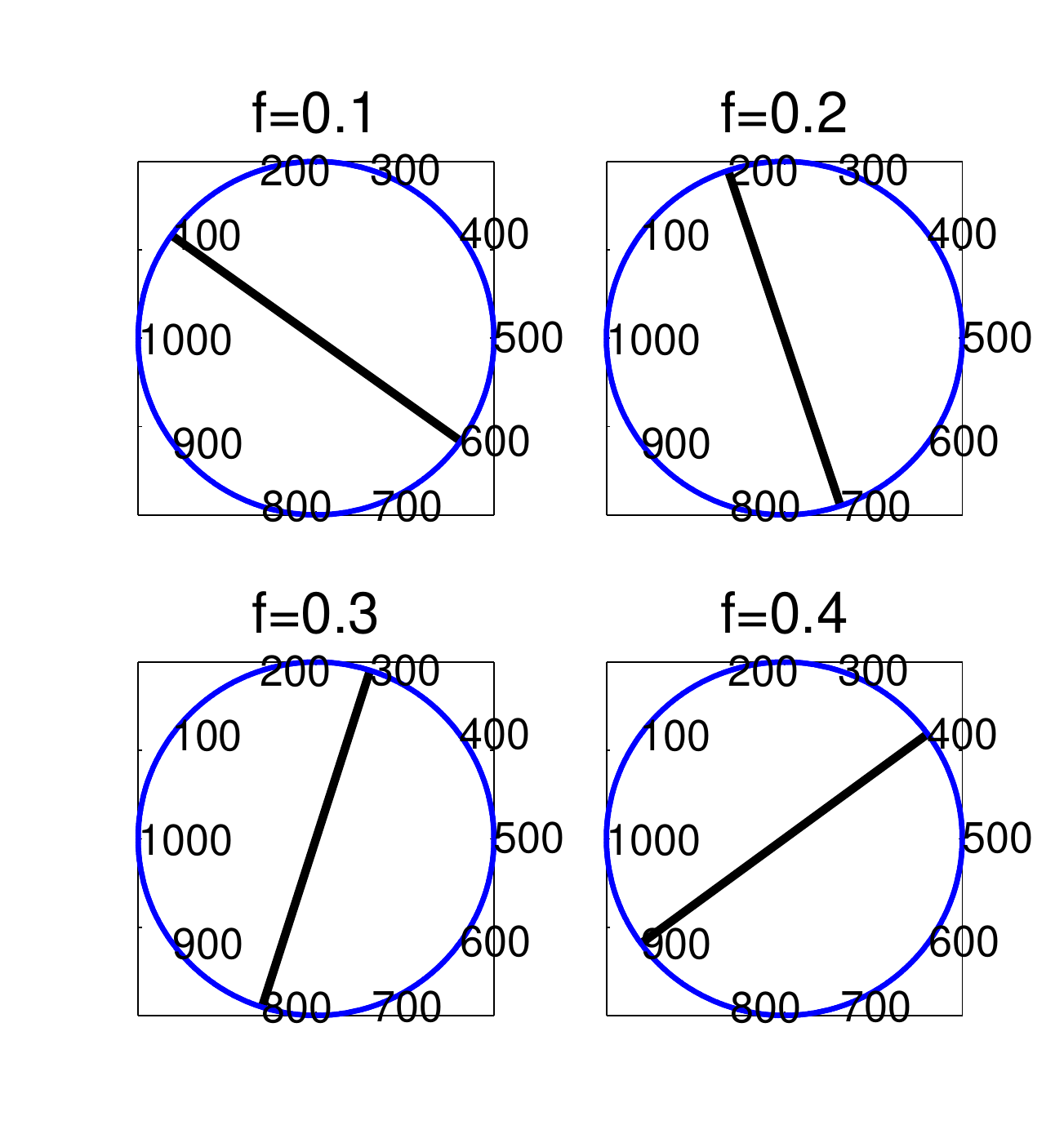}
\includegraphics[height=4.65cm]{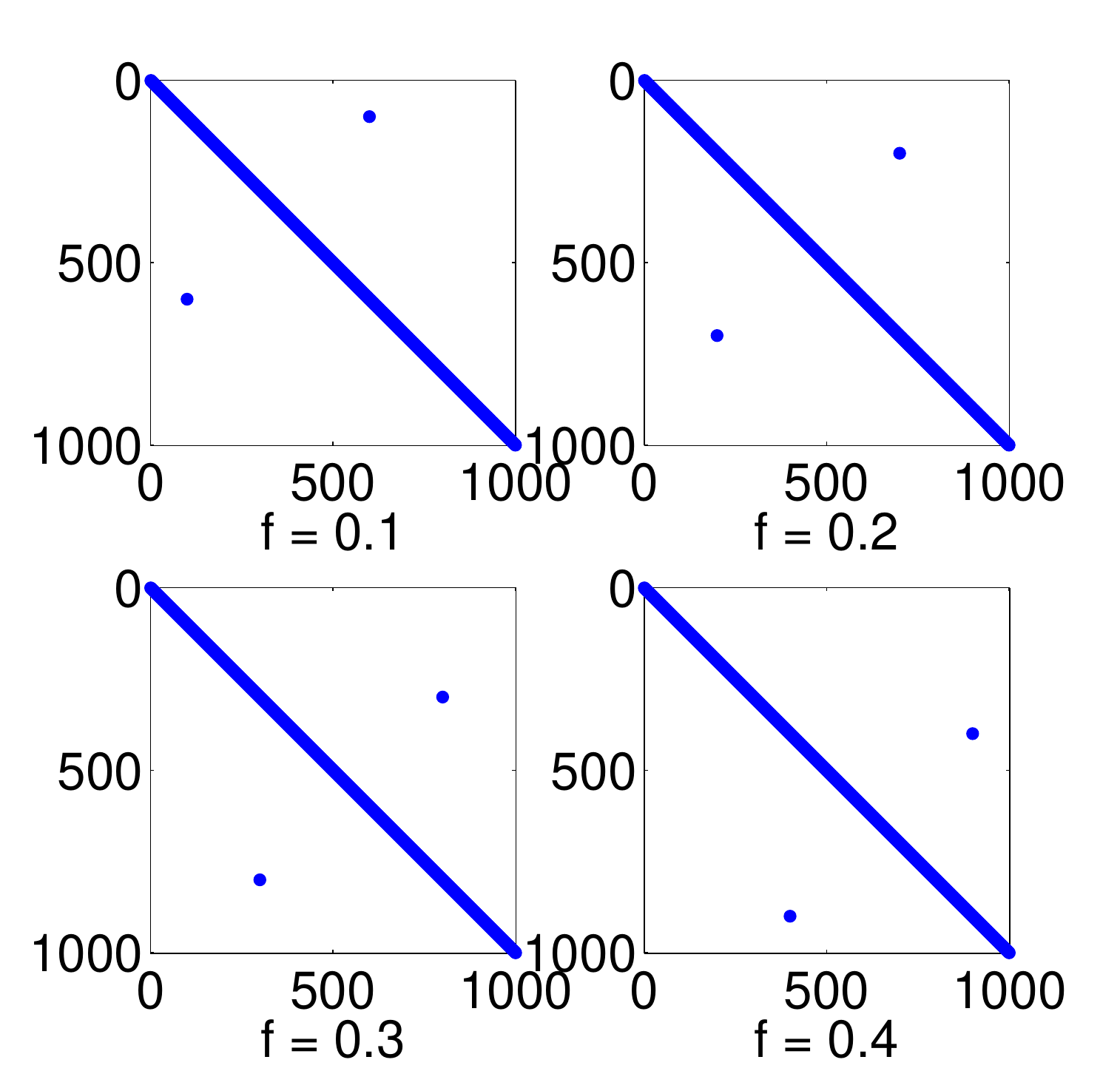}
\end{center}
\caption{Spectral correlation graphs $\mathcal G_{f,\rho}$ for $f=[0.1, 0.2, 0.3, 0.4]$ and correlation threshold $\rho = 0.9$, which corresponds to a false positive probability of $10^{-65}$. The data used here is a set of synthetic time series obtained by band-pass filtering of a Gaussian white noise series with the band-pass filters shown in Fig. \ref{fig:fftMag}. As can be seen, complex correlation screening is able to extract the correlations at specific frequencies. This is not directly feasible in the time domain analysis.}
\label{fig:FreqGraphs2}
\end{figure}

\section{Conclusion}
This chapter presented a spectral method for 
correlation analysis of stationary multivariate Gaussian time series with a focus on identifying correlation hubs. 
The asymptotic independence of spectral components at different frequencies allows the problem to be decomposed into independent problems at each frequency, thus improving computational and statistical efficiency for high-dimensional time series. 
The method of complex-valued correlation screening is then applied to detect hub variables at each frequency.  Using a characterization of the number of hubs discovered by the method, thresholds for hub screening can be selected to avoid an excessive number of false positives or negatives, and the statistical significance of hub discoveries can be quantified.  The theory specifically considers the high-dimensional case where the number of samples at each frequency can be significantly smaller than the number of time series.   
Experimental results validated the theory and illustrated the applicability of complex-valued correlation screening to the spectral domain.

\section{Acknowledgment}
This work was partially supported by AFOSR grant FA9550-13-1-0043.

%
%
 \bibliographystyle{spmpsci}
 \bibliography{Refs}
%


\printindex
\end{document}